\documentclass[11pt]{article}
\usepackage{authblk}

\usepackage[framemethod=TikZ]{mdframed}
\usepackage[linesnumbered,ruled,vlined]{algorithm2e}
\usepackage[many]{tcolorbox}
\usepackage[utf8]{inputenc} 
\usepackage{algpseudocode}  
\usepackage{amsmath}
\usepackage{amssymb}
\usepackage{amsthm}
\usepackage{appendix}
\usepackage{blindtext} 
\usepackage{caption}
\usepackage{empheq}
\usepackage{geometry}
\usepackage{graphicx}
\usepackage{hyperref}  
\usepackage{cleveref}
\usepackage{soul}
\usepackage{subcaption}
\usepackage{tikz}
\usepackage{xcolor}
\usepackage[normalem]{ulem}

\usepackage{newunicodechar}
\hypersetup{
    colorlinks=true,
    linkcolor=blue,
    filecolor=blue,      
    urlcolor=cyan,
    citecolor = cyan,
}

\makeatletter  
\newif\if@restonecol  
\makeatother

\definecolor{dg}{RGB}{0,204,0}
\definecolor{dr}{RGB}{88,0,0}
\definecolor{db}{RGB}{0,0,88}

\newcommand{\Bigprt}[1]{\Big( #1 \Big)}
\newcommand{\bigprt}[1]{\big( #1 \big)}
\newcommand{\brkt}[1]{\left[ #1 \right]}
\newcommand{\set}[1]{\{ #1 \}}

\newcommand{\E}{\mathcal{E}}

\renewcommand{\P}{\mathbb{P}}

\newcommand{\ie}{\textit{i.e.}}

\newcommand{\esps}[2]{\mathbb{E}_{#1}\brkt{#2}}

\newcommand{\etal}{\textit{et al. }}

\newtcolorbox{mybox}[3][]
{colback=white,
colframe=#2!20,
colbacktitle=#2!20!,
coltitle=#2!20!black,
title={#3},
#1
}

\newtheorem{proposition}{Proposition}
\newtheorem{proposition-definition}{Proposition-definition}
\newtheorem{theorem}{Theorem}
\newtheorem{remark}{Remark}

\usepackage{array}
\newcommand{\PreserveBackslash}[1]{\let\temp=\\#1\let\\=\temp}
\newcolumntype{C}[1]{>{\PreserveBackslash\centering}p{#1}}

\providecommand{\keywords}[1]
{
  \small	
  \textbf{\textit{Keywords---}} #1
}
\title{The self-exciting nature of the bid-ask spread dynamics}

\author[1]{Ruihua RUAN\thanks{Corresponding author.
Email: ruan@ceremade.dauphine.fr}}
\author[1]{Emmanuel BACRY}
\author[2]{Jean-François MUZY}
\affil[1]{CEREMADE, CNRS-UMR 7534, Universit\'e Paris-Dauphine PSL\protect\\
Place du Mar\'echal de Lattre de Tassigny, 75016 Paris, France}
\affil[2]{SPE CNRS-UMR 6134, Universit\'e de Corse
BP 52, 20250 Corte, France}

\begin{document}
\maketitle

\begin{abstract}
The bid-ask spread, which is defined by the difference between the best selling price and the best buying price in a Limit Order Book at a given time, is a crucial factor in the analysis of financial securities. In this study, we propose a "State-dependent Spread Hawkes model" (SDSH) that accounts for various spread jump sizes and incorporates the impact of the current spread state on its intensity functions. We apply this model to the high-frequency data from the Cac40 Euronext market and capture several statistical properties, such as the spread distributions, inter-event time distributions, and spread autocorrelation functions. We illustrate the ability of the SDSH model to forecast spread values at short-term horizons.
\end{abstract}
\section{Introduction}
The bid-ask spread, defined  at a given time in a Limit Order Book by the difference between the smallest ask (selling) price and the largest bid (buying) price, is a quantity of great interest for financial securities. It corresponds to the cost of an ``immediate" transaction as respect to a more patient one.
Many  studies in economics literature are devoted to the bid-ask spread which is,
in general, decomposed into order processing costs, adverse selection costs and inventory risk (for the liquidity providers/market makers who earn money from the spread) \cite{glosten1988estimating, stoll1989inferring, huang1997components,glosten1985bid}.   
The spread is often used as a proxy for market liquidity 
a small spread being associated with very liquid markets. For a literature review of liquidity measures, we refer the reader to  \cite{diaz2020measuring,goyenko2009liquidity, fong2017best,fosset2020non}.  Let us also mention that the bid-ask spread has been proved to be very closely related to realized volatility. Indeed, many empirical studies \cite{goyenko2009liquidity,fong2017best} tend to prove that these two quantities are highly positively correlated \cite{bessembinder1994bid,zumbach2004trading,wyart2008relation,dayri2015large}.
Many other statistical properties of the bid-ask spread have been the focus of various works. For instance, it has been shown that spread has a distribution that is fat-tailed and its dynamics is characterized by a long range (power-law) auto-correlation function \cite{mike2008empirical,plerou2005quantifying,bouchaud2018trades,bouchaud2009markets,gross2013predicting, fall2021forecasting}. In refs
\cite{mike2008empirical,ponzi2006market,zawadowski2006short}, it is shown that after a large variation of the spread (i.e., a temporary liquidity crisis), the spread decays slowly back to an equilibrium value. Let us note that, though, most of the time, these statistical properties are obtained through direct empirical studies on historical spread time-series, some papers tackle the statistical properties of the spread (mainly the distribution of the spread values) via some statistical models of the limit and market order flows
\cite{bouchaud2002statistical,daniels2003quantitative,smith2003statistical,foucault2005limit,rocsu2009dynamic,abergel2013mathematical,muni2017modelling}.
The reader will find a good overview of the different statistical properties of the bid-ask spread in the chapter 7 of  \cite{bouchaud2009markets}. 


In order to understand the properties of market prices and their formation, 
in the context of electronic markets, many approaches involving point processes have been proposed to describe the occurrence of order book events (see e.g. \cite{smith2003statistical,Cont10, abergel2013mathematical}) .  Within this context, as reviewed in \cite{BacryRev2015}, Hawkes processes is a very popular class of models that has proven very efficient to describe the dynamical properties of different quantities like the market activity \cite{bowser07,toke2012,hardiman,morariu2018state}, the mid-price \cite{b13}, the best bid / best ask prices \cite{Kyungsub22} or the first (L1) book levels \cite{large07, bacry2016estimation,zheng2014modelling}. 
Hawkes processes constitute a class of multivariate point processes that were introduced in the seventies by A.G. Hawkes \cite{hawkes1971point} notably to model the occurrence of seismic events.  They involve an intensity vector that is (in its original form) a simple linear function of past events. This popularity of Hawkes processes can be explained above all by their great simplicity and flexibility. 
Our main purpose in this paper is to study whether one can design a dynamical model for 
spread fluctuations based on Hawkes processes. Though many studies can be found in the literature about the bid-ask spread and its statistical properties, very few offer  models for its high-frequency dynamics.  
Besides, some recent econometric approaches involving long-memory auto-regressive, Poisson point processes \cite{gross2013predicting,CATTIVELLI2019}, one can mention few models that rely on Hawkes processes. In 2014, Zheng \etal proposed in \cite{zheng2014modelling} one of the first model for the bid-ask spread dynamics of a financial asset. It uses a constrained 2-dimensional Hawkes process. The spread $S_t$ is seen as a strictly positive integer value (a multiple of the tick value, i.e., the minimum increment defined by the market between two quotation values). The first (resp. second) dimension of this Hawkes process $S^+_t$ (resp. $S^-_t$) is used for coding the positive (resp. negative) jumps of the spread, it is an increasing function of time that jumps by 1 whenever a positive jump of the spread occurs. 
In this model all the jumps of the spread are considered to be of size 1 tick.  
Thus, one gets $S_t = S^+_t-S^-_t$. The constrained Hawkes process is thus defined by the following expression of the intensity $\lambda^+$ (resp. $\lambda^-$) of the jump process $S_t^+$ (resp. $S_t^-$) : 
\begin{equation}\label{spread0}
\begin{split}
\lambda^+_t &= \mu^+ + \sum_{e\in\{+,-\}}\int_0^t \phi^{+,e}(t-s) dS_s^e\\
\lambda^-_t &= 1_{S_{t-}\geq 2}(\mu^- + \sum_{e\in\{+,-\}}\int_0^t \phi^{-,e}(t-s)dS_s^e)\\
\end{split}
\end{equation}
where $\mu^+$ (resp. $\mu^-$) corresponds to a  constant exogenous intensity (a term that does not depend on the past events)  
and the 4 functions  $\{\phi^{e,e'}(t)\}_{e,e' = \pm}$ are kernel causal (i.e., support in $[0,+\infty[$) functions that encode the (endogenous) influence of past jumps (of type $e'$) on the occurrence of future jumps (of type $e$).  Let us point out that for estimation purpose, the kernels are often taken to be exponential functions (or sum of exponential functions). Indeed, 
for such functions the likelihood is explicit and can be computed in a fast way. In \cite{zheng2014modelling}, simple exponential kernels were chosen :  
$\phi^{e,e'}(t) = \alpha^{e,e'} \beta e^{-\beta t}$ ($\beta$ is the same for all the kernels).
Let us point out that, Zheng \etal introduced a non-linearity (the term $1_{S_{t-}\geq 2}$)  in the "classical" definition of a Hawkes process. This non-linearity is necessary to constrain the spread to keep from taking negative or zero values. In their work, Zheng \etal performed numerical estimation of their model (using Maximal Likelihood Estimations) and studied some statistical properties of their model. 
In a more recent paper \cite{fosset2020endogenous}, the authors build a simplified version of the previous model and focus on the relation between spread dynamics and liquidity crises. They introduce the constrained 2-dimensional Hawkes process defined by
\begin{equation}\label{spread1}
\begin{split}
    \lambda^+_t &= \mu^+ + \int_0^t \alpha\beta e^{-\beta(t-s)}dS_s^+\\
    \lambda^-_t &= \mu^- 1_{\{S_t\geq2\}}
\end{split}
\end{equation}
where the notations are exactly the same as in model (\ref{spread0}). They use $\alpha_c = 1-\frac{\mu^+}{\mu^-}$ to distinguish different regimes. When $\alpha<\alpha_c$, the Hawkes system is stable and the invariant distribution of spread is given. When $\alpha_c < \alpha < 1$, the model is still stable but the spread has a linear growth. When $\alpha>1$, the system is explosive and there is a liquidity crisis.


In this paper we propose an approach which is inspired by these two papers but that will focus more on fitting
empirical features and matching observed properties from market data. More specifically,
we propose a new ``State Dependent Spread Hawkes'' (SDSH) model in order to account for the bid-ask spread fluctuations. Our model can be seen as a generalization of the two previously introduced models in the sense that it is a 2$K$-variate ($K$ possibly greater than 1) Hawkes process that is able to account for the different jump sizes (up to size $K$). It includes enough kernels to encode the influence of the various past jumps on the various future jumps. Each component includes a non-linearity component, as in the previous models, in order to ensure strictly positive values for the spread. In order to allow more complex dynamics, each kernel is written as a sum of several exponential functions (and not only one, as in the previous models). 
Our ambition is  to capture all the aspects of the dynamics of the bid-ask spread that were not captured by the previous models. 

The paper is organized as follows. 
The next section (Section \ref{sec:model_sec2}) is devoted to a detailed definition of our new model for the bid-ask spread. The Markov and ergodicity properties of this model is studied in a particular case. 
Section \ref{sec:simu_estim} illustrates this model with some numerical simulations and presents the estimation procedure on these simulations. Section \ref{sec:res}, the core of our paper, is devoted to empirical results. This sections is divided in 4 subsections. The first two ones correspond to a quick presentation and basic statistical properties of the real financial time-series that will be used all along the section. The third one is devoted to estimation of our model on these data. Parameters estimation (including the kernels and the constraint parameters) are discussed thoroughly. The last one is devoted to the goodness of fit of the model through different quantities, mainly : the inter-event time distribution, the spread distribution and the autocorrelation function. Section \ref{sec:pred} provides a first approach that demonstrates the interest of the SDSH model in the issue of short-term forecasting of spread values. We conclude in section \ref{sec:conclusion}.   

\section{Modeling spread dynamics using State Dependent Spread Hawkes processes}\label{sec:model_sec2}

\subsection{Notations}
Let us consider a limit order book (LOB) associated with a given asset. We note 
$S_{t}$ the bid-ask spread of this order book at time $t$, i.e., the difference, at time $t$, of the best ask price and the best bid price. We choose to express $S_t$ in {\em tick} units (the {\em tick} corresponds to the smallest price increment authorized by the market). $S_t$ is therefore a right continuous process that can  take only integer (strictly positive) values (i.e., $S_t\in \mathbb{N}^*$). Its smallest possible value is $S_t=1$ (tick). 

The process $S_t$ can be seen as a jump process and various orders sent to the LOB, depending on their type and size may induce jumps on $S_t$ of various sizes. In our model, an event corresponds to a jump of $S_t$ of a given size. In practice, very large jumps almost never occur, so with no loss of generality, we can limit the set of events to $\mathcal{E} = \{+1, +2,...,+K, -1, -2,...,-K\}$ (where $K$ is a hyper parameter of the model to be fixed for each asset).  

For each event $e \in \mathcal{E}$ (i.e., for each jump size $e$), we denote by $S^e_t$ the counting process that counts (along time, starting arbitrarily, at time 0) the number of jumps of size $e$ that occurs. Thus, the spread process $S_t$ can written as:
\begin{equation}
\label{eq:S}
S_t = S_0 + \sum_{k=1,2,...,K}kS^{+k}_t - \sum_{k=1,2,...,K}kS^{-k}_t.
\end{equation}
In real limit order books, it is clear that the dynamics of the spread jumps highly depend on the size of the jumps, it is thus natural to build a model that allows different dynamics for different jump sizes. 

\subsection{The SDSH spread model}
\label{sec:model}
Our model for spread dynamics basically consists in considering that the multivariate counting process $\{S_t^{e}\}_{e \in \mathcal{E}}$ follows a $2K$-variate Hawkes process. However, since the dynamics of the various events depend on the state of the current spread $S_t$ itself (the least being that when the spread is 1 tick, no negative event can occur), we introduce, in the classical Hawkes framework, a term that depends on $S_t$, i.e., a {\em state variable} that accounts for the current size of the spread. More precisely, if we note $\lambda^e_t$ the conditional intensity (at time $t$) of the counting process $S^e_t$, the State Dependent Spread Hawkes spread model writes:  
\begin{equation}
\label{model} 
\lambda^e_t = f^{e}(S_{t-}) 
\left[
\mu^e + \sum_{e'\in\mathcal{E}} \int_0^t \phi^{e, e'}(t-s) dS^{e'}_s
\right] 
\end{equation}
where (following the classical Hawkes processes framework) 
\begin{itemize}
\item $\mu^e$ is the exogenous base intensity for the jumps of size $e$ 
\item $\Phi(t) = \{\phi^{e,e'}(t)\}_{e,e' \in {\cal E}}$ is the Hawkes kernel matrix whose elements are the Hawkes (positive valued) interaction kernels, encoding the influence of past jumps of size $e'$ on future jumps of size $e$. 
\item we assume that these kernels are parameterized as a sum $L$ exponentials, i.e.,
\begin{equation}
    \label{eq:phi}
    \phi^{e,e'}(t)=\sum_l^L\alpha_l^{e,e'}\beta_l e^{-\beta_l t},
\end{equation}
(this is hardly restrictive since many behaviour can be easily reproduced by a sum of exponentials \cite{bochud2007optimal})
\end{itemize}
and where the current size of the spread is taken into account through the global multiplicative term
$f^e(S_{t_-})$ where 
\begin{itemize}
\item $S_{t^-}$ is the the left limit of $S$ at time $t$ (i.e., the spread size "just before" time $t$), 
\item $f^e$ is a non-negative function defined on $\mathbb{N}^*$, with the constraint that $f^{-k}(n)=0$ when $n-k\leq 0$ (in order to prevent jumps that would lead to a non strictly positive value for the spread). Let us note that the Equation \eqref{eq:phi} does not change if we multiply $f^e$ by a factor and divide $\mu^e$ and the $\phi^{e,e'}$ by the same factor. So, in the following,  we will fix arbitrarily the first non zero value of  $f^e(s)$ to be 1, i.e., $f^e(\min_s\{s,~f^e(s) \neq 0\}) = 1$.
\end{itemize}
Let us point out that, in order to account for the impact of the current spread size on the different dynamics of the spread jumps, several choices could have been made. 
Maybe, the most natural one would have been to make the kernel themselves depend on the spread size $S_{t_-}$. However, such a choice would lead to increase drastically the number of parameters needed to encode the kernels themselves, leading to important estimation instabilities.
The choice we made allows several things at once 
\begin{itemize}
    \item keeping the number of estimation parameters low
    \item encoding in an easy way the fact that negative values for $S_t$ are not allowed
    \item allowing the model to account for the well known fact that the spread dynamics is mainly mean reversing and, the higher the spread size is, the higher is the probability for large negative jumps to occur. 
\end{itemize}

\subsection{Markov property and ergodicity} 
\label{sec:markov}
Using Eq. \eqref{eq:S}, it is easy to prove that, since the Hawkes kernels are sums of exponential functions,  then the following property holds:
\begin{proposition}
The process $(S_t, X_t)$, where $X^{e,e'}_t := \int_0^t  \phi^{e,e'}(t-s) dS_s^{e'}$,  is a Markov process
\end{proposition}
In the simple case where $K = 1$ (i.e., only jumps of size +1 or -1 are allowed) and $L = 1$ (i.e., only one exponential kernel), one can prove (see Annex \ref{annex:ergo}) the following  ergodic property :  
\begin{proposition}
\label{thm_ergodicity}
Assume $K = 1$, i.e., ${\cal E} = \{-1,1\}$ and $L=1$, i.e,  $\phi_{e,e'}(t) = \alpha^{e,e'}\beta e^{-\beta t}$,\\
If the  following conditions are satisfied : 
\begin{empheq}[left={(A)=}\empheqlbrace]{align}
	&f^-(1) = 0 \tag{{A\textsubscript{1}}}\\
	&f^{-}(S) \geq \gamma S \text{ for some $\gamma > 0$ when $S\geq2$} \tag{{A\textsubscript{2}}}\\
	&\sup_{S\geq 1}\{f^+(S)\}(\alpha^{+,-}+\alpha^{+,+}) < 1 \tag{{A\textsubscript{3}}}
\end{empheq}	
then the process $(S_t,X_t)$ is a V-uniformly ergodic Markov process.
\end{proposition}
Thus, under these assumptions, the spread process has some stationary distributions.  Let us notice that condition (A\textsubscript{2}) is a condition ensuring a "return to the mean value" analog to the condition of a proportional cancellation rate in \cite{smith2003statistical,abergel2013mathematical, wu2019queue}. 
The condition (A\textsubscript{3}) corresponds to a stability condition ensuring that the number of upward spread jumps does diverge exponentially in a finite interval.
This is a first result, the general proof (for any $K$ and $L$) seems to be much more difficult (see Remark \ref{ergo2}), it will be the focus of a forthcoming paper.

\section{Numerical simulation of the SDSH spread model and parametric estimation}\label{sec:simu_estim}

\subsection{Simulation}
\label{sec:simulation}
In order to perform numerical simulations, we use the classical "thinning method" introduced by \cite{lewis1979simulation, ogata1981lewis} and implemented in the {\em tick} open source library \cite{bacry2017tick}.

To make it easier for readers to understand how our model works in practice, Figure \ref{fig:model_intensity} shows the result of such a simulation during the first 20 seconds. We chose the following parameters 
\begin{itemize}
    \item $K=1$, i.e., $\mathcal{E}=\{+1,-1\}=:\{+,-\}$
    \item $L=1$ and $\phi^{e,e'}(t) = \alpha^{e,e'}e^{-\beta t}$, where $\beta = 1$, $\alpha^{+,+} = \alpha^{-,-} = 0.1$ and  $\alpha^{+,-} = \alpha^{-,+} = 0.2$
    \item $\mu^+ = \mu^- = 0.3$
    \item $f^+(1)=1, f^+(2)=0.7$ and  $f^+(S)=0.3$ for $S\geq 3$,
    \item $f^-(1)=0, f^-(2)=1$, and $f^-(S)=5$ for $S\geq 3$.
\end{itemize}


\begin{figure}[h]
     \centering
    \includegraphics[width=0.9\textwidth]{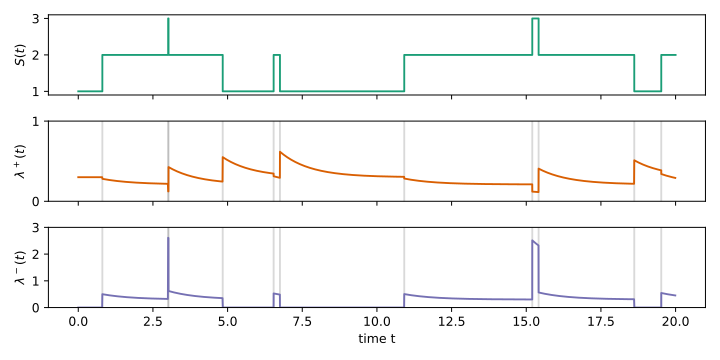}
    \caption{A realization of the SDSH spread model during 20 seconds using parameters as indicated in the text. As we see, since we chose $f^{-}(S)$ to be very large (5) as soon as $S\ge 3$, the spread does not stay long at a value of 3.}
    \label{fig:model_intensity}
\end{figure}

\subsection{Estimation principles}  
\label{sec:estimation}
For parametric estimation, we will use again the {\em tick} python open source library which natively allows robust parametric Maximum Likelihood Estimation (MLE) estimation of multi-dimensional Hawkes processes using sum of exponentials. We slightly tweaked the algorithms of this library in order to take into account the $f^e(S)$ terms. The algorithm used in this model closely follows the one used for the QRH-II model described in \cite{wu2019queue}. The likelihood function for this model is located in Appendix \ref{mle}.


More precisely, prior to the parametric simulation, a certain number of hyper-parameters have to be fixed (according to the specificity of the corresponding assets, see Section \ref{setting} for specific cases) : 
\begin{itemize}
    \item $K$ : the highest jump size allowed by the model. This hyper-parameter should be carefully chosen after looking at the effective probability of occurrence of various jump sizes on the real data,
    \item $L$ and $\{\beta_l\}_{l \in L}$ : in practice choosing the $\beta_l$ so that they are logarithmically spaced is sufficient to recover a very large variety of behavior, i.e., $\beta_l = \beta_1 10^{l-1}$. $L$ and $\beta_1$ should be chosen to adapt to the timescale range one is interested in. 
    \item each $f^{e}(s)$ is a function of $s \in \mathbb{N}^*$, thus corresponds to an infinite number of parameters. For parametric estimation to be feasible, we have somehow to make some assumptions on the behavior of the  $f^{e}(s)$ functions for large $s$. Very large spread are extremely rare events, so basically any assumption will work, as long as it enforces mean reversion of the spread. Thus, we arbitrarily chose to consider that all the functions $f^e(s)$ are constant for $s$ greater that a fixed value $\bar S$. In practice, the value of $\bar S$ should be chosen so that spread of size greater than $\bar S$ should be extremely rare events. 
\end{itemize}
Given these hyper-parameters, our MLE parametric estimation algorithm allows to estimate 
\begin{itemize}
    \item the exogenous intensities $\{\mu^e\}_{e \in {\cal E}}$ (so a total of $2K$ parameters)
    \item the kernel parameters $\{\alpha_l^{e,e'}\}_{l \in L, (e,e') \in {\cal E}^2}$ (so a total of $4LK^2$ parameters)
    \item the values $\{f^e(s)\}_{e \in {\cal E}, s \in [1..\bar S]}$ (so a total of ($2K\bar S - \frac{K^2+K}2 -2K$) parameters, as long as $K<\bar S$ (let us remind that we fixed arbitrarily the first non zero value of $f^e(s)$ to be 1.)
\end{itemize}
So that amounts to $2K+4LK^2+2K\bar S - \frac{K^2+K}2 - 2K$ parameters.

In the following we illustrate the estimation procedure on simulated data.

\subsection{Estimation on simulated data}  
In this example, we simulated our model in dimension 2, i.e., for 
$\E = \set{+1,-1} := \set{+,-}$ : the only possible moves for the spread are upward or downward moves of one tick. 
Simulation (using the thinning method \cite{lewis1979simulation,ogata1981lewis} as explained in Section \ref{sec:simulation}) was run to produce 50 samples of size 5000 seconds using the following parameters.
We set the parameters for the simulation to be (using the notations introduced in Section \ref{sec:model}) : 
\begin{itemize}
    \item $\mu^+ = 0.3, \mu^- = 0.2$
    \item $L = 2$ with $\beta_1 = 20s^{-1}, \beta_2=200s^{-1}$ and the Hawkes kernel is 
    $$\phi(t) = \begin{pmatrix}
    2 & 6\\
    10 & 0
    \end{pmatrix}e^{-20t}+
    \begin{pmatrix}
    4 & 20\\
    20 & 4
    \end{pmatrix}e^{-200t}$$
    \item $f^+(1)=1, f^+(2) = 0.8, f^+(3)=0.5, f^+(s)=0.2$ for $s\geq 4$\\
    $f^-(1)=0, f^-(2) = 1, f^-(3)=2, f^-(s)=3$ for $s\geq 4$
\end{itemize}
\textbf{Estimation}
Then we estimate the parameters using $\beta_1 = 10s^{-1}, \beta_2=100s^{-1}, \beta_3=1000 s^{-1}$, i.e.
$$\phi(t) = 10\begin{pmatrix}
    \alpha_1^{++} & \alpha_1^{-+}\\
    \alpha_1^{+-} & \alpha_1^{--}
    \end{pmatrix}e^{-10t}+
    100\begin{pmatrix}
    \alpha_2^{++} & \alpha_2^{-+}\\
    \alpha_2^{+-} & \alpha_2^{--}
    \end{pmatrix}e^{-100t}+
    1000\begin{pmatrix}
    \alpha_3^{++} & \alpha_3^{-+}\\
    \alpha_3^{+-} & \alpha_3^{--}
    \end{pmatrix}e^{-1000t}$$

\begin{figure}[ht]
    \centering
    \includegraphics[width=\textwidth]{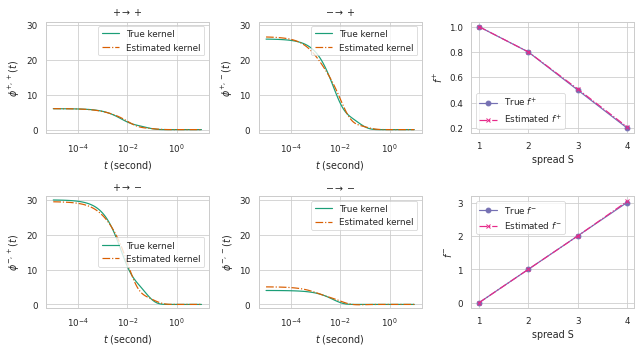}
    \caption{Numerical estimation from samples of a SDSH spread model. The 4 figures on the left show the true kernels and the estimated kernels, and the 2 figures on the right show the true $f$ values and the estimated $f$ values. The estimated $\mu^+$ is 0.29 and the estimated $\mu^-$ is 0.19.}
    \label{fig:test1}
\end{figure}

As it can be seen in Fig. \ref{fig:test1}, our estimation method allows one to recover the parameter values $\mu^e$ and to reproduce very well both the functions $f^e(S)$ and the Hawkes kernel shapes. It is worth noting that the basis exponential functions for the simulation are very different from those for estimation.
\section{Empirical results}\label{sec:res}
\label{setting}
\subsection{Data}
In this section,  we calibrate the SDSH spread model using high-frequency data from the Cac40 french Euronext Market. The data correspond to every single change of the spread (characterized by the size of the change and a time stamp with a precision of $1\mu s$) for 3 stocks, namely AXA, BNP and NOKIA and for the CAC40 index Future.  The data for stocks (resp. Cac40 Future) are extracted from February 1st 2017 (resp. January 4th 2016) to February 28th 2018 (resp. February 28th 2017). 
In order to minimize the intraday seasonality of the data we used only the data in the intraday slot [10am,12am]. Thus, for a given asset, the estimation of the model will be performed considering each day as an independent realization. Thus, in order to avoid days with too short time series, we took out, for a given asset, all the days associated with a number of events (i.e., spread changes) below a threshold. 
We refer the reader to the Table \ref{stats_data} for getting some basic information on each of these data series. Let us point out that the Cac40 Future has a much larger tick size than the 3 other assets, we expect the jump size for the spread to be much smaller on average. 
\begin{table}[!h]
\centering
\setlength\tabcolsep{1.5pt}

\begin{tabular}{C{30mm}C{23mm}C{23mm}C{20mm}C{26mm}C{20mm}C{15mm}}
\hline
Asset & tick size & Min. \#events &  \#days &   Total \#events & $\esps{event}{S}$ & $\esps{cal}{S}$\\
\hline
 CAC40 Future  &  0.5 & 5,000 &  100 &  1,026,156& 1.51 & 1.40\\
 AXA  & 0.005 & 3,000 &  130 &  617,309 & 2.44 & 3.04\\
 BNP  & 0.01 & 5,500 &  100 &  1,079,464 & 2.54 & 3.17\\
 NOKIA  & 0.001 &2,000 & 108 &  570,754 & 3.52 & 4.43\\
\hline
\end{tabular}
\caption{Characteristics of the data used for model estimation. For each asset, for each day we only consider the slot [10am,12am]. Moreover, we keep only days where the number of events (i.e., the number of times the spread changes) is above the  'Min. \#events' number. $\esps{event}{S}$ and $\esps{cal}{S}$ are the expectations of spread with two different distributions : event time distribution and calendar time distribution. See Eq. \eqref{eq:event-dist} and \eqref{eq:calendar-dist} for their definitions. }
\label{stats_data}
\end{table}
\subsection{Spread distribution and hyper-parameter settings}
\label{sec:hypersettings}
In order to choose the the value for hyper parameters, following the lines of Section \ref{sec:estimation}, we have to study the empirical distribution of the spread.
Inspired by \cite{bouchaud2009markets}, we can consider two ways to measure the spread distribution. They are both constructed by an average over daily distributions which can be obtained in two ways :  
\begin{itemize}
\item Calendar time daily distribution : 
\begin{equation}
    \label{eq:calendar-dist}
\mathbb{P}_{cal}(spread=S) = \frac{1}{T}\int_0^T 1_{S_u=S}du
\end{equation}
(where $[0,T]$ corresponds to the slot 10am-12am) 
\item Event time daily distribution  : 
\begin{equation}\label{eq:event-dist}
\mathbb{P}_{event}(spread=S) = \frac{1}{N}\sum_{n=1}^N 1_{[S_{t_n-}=S]}
\end{equation}
(where $N$ is the total number of events on a given day). 
\end{itemize}
The solid lines in Figure \ref{fig:spread_dist} illustrates the so-obtained results. Let us point out that it also displays (the right column) the distribution of the variation of the spread, i.e., $dS$.

\paragraph{Choice of K.} As expected, since the tick of the Cac40 Future is much larger than the other assets, we expect the values of the spread variations to be mainly of 1 tick. One can see that this is indeed that case in the last figure \ref{fig:future_spread_dist}: the amplitude of spread variations $dS$ is never larger than 1 tick.  Consequently, it seems  natural to choose $K=1$ for this asset. If we follow the same guidelines (i.e., choosing the value of $K$ as the minimum value so that the events corresponding to a change of the spread of $K$ ticks is very rare) we see that choosing $K=2$ for all the other assets seems reasonable (see Figure \ref{fig:axa_spread_dist}, \ref{fig:bnp_spread_dist}, \ref{fig:nokia_spread_dist}). 

\paragraph{Choice of $\bar S$.} As explained in Section \ref{sec:estimation}, choosing arbitrarily $\bar{S}$ as the minimum value for which the probability of the spread to be this value is not extremely close to zero seems reasonable (there is no way we have enough information to perform reliable estimation of $f^e(S)$ for values of $S$ which hardly never occur). 
More precisely, we chose $\bar{S}$ to be the maximum $s$ for which $\mathbb{P}_{event}(Spread=s) \geq 1\%$.
This leads to the following values : $\bar S = 5$ (AXA), $\bar S = 5$ (BNP), $\bar S = 8$ (Nokia) and  $\bar S = 2$ (CAC40 Future).

\paragraph{Choice of $L$.} As we will see (see Section \ref{sec:hawkes_kernel}), choosing $L=6$ and $\beta_1 = 10^{-1}s^{-1}$ for all the assets (consequently, following Section \ref{sec:estimation} $\beta_2 = 1s^{-1}$, $\beta_3 = 10s^{-1}$, $\beta_4 = 10^2s^{-1}$, $\beta_5 = 10^3s^{-1}$, $\beta_6 = 10^4s^{-1}$) is sufficient to capture the kernels dynamics on the time-scale $[10^{-4}s,10s]$. Let us point out that we performed estimation with largest values of $L$ and smaller values of $\beta_1$, it does note change significantly the results (while it increases significantly the duration of the estimation procedure or alternatively leads to unstable results). 


Table \ref{tab:hyperparam} summarizes all the choices for the hyper parameters
\begin{table}[!h]
\centering
\begin{tabular}{lccccc}
\hline
Asset  & $K$ & $\bar{S}$ & Kernel time scale & \# parameters \\
\hline
 AXA   & $2$ &  5 & $10^{-4}s\to 10^1s ~(L=6)$&  113\\
 BNP   & $2$ &  5 &  $10^{-4}s\to 10^1s ~(L=6)$  & 113\\
 NOKIA  & $2$ & 8 &  $10^{-4}s\to 10^1s ~(L=6)$  & 125\\
 CAC40 Future   &  $1$ &  2 & $10^{-4}s\to 10^1s ~ (L=6)$ &  27\\
\hline
\end{tabular}
\caption{Values chosen for the hyper parameters (following the guidelines in Section \ref{sec:estimation}) for each dataset corresponding to each asset. $K$ corresponds to the maximum jump size for the spread. $\bar S$ to the abscissa above which the functions $f^e(s)$ are considered as constant and the Kernel time scales is deduced from the choice of $L$ and $\beta_1$}
\label{tab:hyperparam}
\end{table}

\subsection{Estimation}
We performed parameters estimation on each time-series described in the previous Section using MLE following the guidelines of Section \ref{sec:estimation}. As explained in this Section, all the estimations in this paper were performed using  the {\em Tick} open-source package \cite{bacry2017tick}, after having adapted the MLE exponential-kernel estimation algorithm of {\em Tick} to account the state dependent function $f^e(S)$. In the following we will show the various results of this estimation and comment on them, starting with the estimation of the $\{f^e(s)\}_{e \in {\cal E}}$ functions and  then of the kernels $\{\phi^{e,e'}(t)\}_{e,e' \in {\cal E}}$ themselves. 

\subsubsection{Estimation of the $\{f^e(s)\}_{e \in {\cal E}}$ functions}
The results of the estimation procedures for the different assets are reported on Fig. \ref{fig:f} (let us recall that, for each $e$, the first positive value of $f^e(S)$ is arbitrarily set to 1). 

We first remark that  all estimated curves display the same behavior for all assets.
As expected, we see that the the $f^e$ functions corresponding to positive events (i.e., upward jumps, $e \in \{+1,+2\}$) are decreasing functions that decrease pretty quickly towards a value close to 0 whereas the $f^e$ functions corresponding to negative events (i.e., downward jumps, $e \in \{-1,-2\}$) are rapidly increasing functions.
When the spread is high, it is clearly pressed downward by inhibiting the positive events  (with small $f^{+1}(s)$ and $f^{+2}(s)$ for large spread), and by exciting the negative events  (with large $f^{-1}(s)$ and $f^{-2}(s)$ values).

The two functions $f^{-1}$ and $f^{-2}$ are clearly saturating though saturation seems to appear faster on $f^{-1}$ than on $f^{-2}$. 
Let us recall that one cannot perform estimation for larger spread than the chosen $\bar S$ due to the lack of statistics (i.e., spread values above the chosen $\bar S$ are extremely rare events). 

Let us remark that, as expected, the so-obtained range for the values of the function corresponding to negative events is in line with the corresponding tick size (and the average spread size) of each asset (see Table \ref{stats_data}). Indeed, the largest the tick size (the smallest the average spread) the smallest the range of values : going from CAC40 Future (being the largest tick size) to NOKIA (being the smallest tick size).

\begin{figure}[h]
     \centering
     \includegraphics[width=\textwidth]{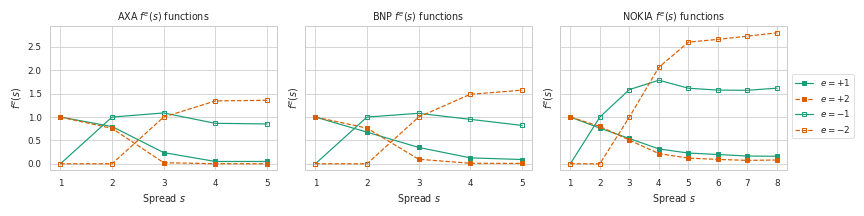}
     \caption{Estimations of the $\{f^e (s)\}_{e \in {\cal E}}$ functions for AXA (left), BNP (middle), NOKIA (right), ${\cal E} = \{-2,-1,+1,+2\}$. As for CAC40 Future, the $\Bar{S}$ is 2 and the $K$ is 1, therefore $f^{-1}(1) = 0, f^{-1}(s) = 1 \text{ if } s\geq 2$, $f^{+1}(1) = 1, f^{+1}(s) = 0.0104 \text{ if } s\geq 2$.
     For each $e$, the first positive value of $f^e(S)$ is arbitrarily set to 1.}
     \label{fig:f}
\end{figure}

\subsubsection{Estimation of Hawkes kernels $\{\phi^{e,e'}(t)\}_{e,e' \in {\cal E}}$} 
\label{sec:hawkes_kernel}
First, let us point out that our estimation procedure essentially leads to  positive valued kernels. Let us recall that the MLE procedure described in Section \ref{sec:estimation} is able to reveal inhibition behavior (i.e., significantly negative valued kernels) when present in the signal. Actually, if one zooms some of the kernels, one would reveal some negative values for a few kernels, but their absolute values are very small and not significant.
Of course, this does not mean that there is no inhibition behavior in the spread counting process. Actually the inhibition behavior are extremely strong but they are taken care, for each component $S^e_t$, 
by the multiplicative term $f^{e}(S_{t^-})$. Recall, this is the main reason why we initially introduced them in our model ("hard" inhibition of some components to prevent negative spread values). 

\vskip .3cm 
{\bf Comparison of kernel integrated quantities} 

Let emphasize that, the introduction of multiplicative terms in the ``State Dependent" Hawkes model (i.e., the fact that our model is not a plain vanilla Hawkes model due to the role of the state variable $S_t$) keeps us from regular interpretations of the $L^1$ norm  of the Hawkes kernels. Indeed, in a vanilla Hawkes model, one generally compare the different values of the $L^1$ norms $\{||\phi^{e,e'}(t)||_1\}_{e,e' \in {\cal E}}$ and uses the classical populationnal interpretation of a Hawkes model ($||\phi^{e,e'}(t)||_1$ represents the average number of events of type $e'$ "directly" generated by an event of type $e$) in order to disentangle the overall dynamics of the process. 
In our case, this interpretation is not possible. Not only it is conditioned to a spread value but the comparison between two norm values  $||\phi^{e_1,e'}(t)||_1$ and  $||\phi^{e_2,e'}(t)||_1$ for two different components $e_1 \neq e_2$ does not make sense.

However, one can still compare the influence of all the past events of a given type $e'$ on the occurrence of an event of a given type $e$.
Indeed, given a  spread jump of size $e$ at time $t_k$. Thus, the spread jumps from the value $s$ (at time $t_k^-$) to $s+e$ (at time $t_k^+$), where $s = S_{{t_k}^-}$.
Then one could study the influence of each endogenous term in the sum (within the brackets) in \eqref{model}, thus comparing the relative values $\int_0^{t_k^-} \phi^{e,e'}(t-s)dS_s^{e'}$ for different $e'$ in order to understand which event type $e'$ is influencing the most the occurrence at time $t_k$ of the jump of size $e$. For that purpose, let us introduce the quantity 
$$
I_{s, e, t_k}(e') :=\int_0^{t^*} \phi^{e,e'}(t^*-u)dN^{e'}_u,~~~\forall e'
$$
for a jump of size $e$ occurring at time $t_k$ and leading to a spread value of $s$. We can then define the corresponding  averaged values
$$
\bar I_{s, e}(e') := \cfrac{\sum_{t_k:dS_{t_k}=e, S_{t_k-}=s} I_{s, e, t_k}(e')}{\#\{t_k:dS_{t_k}=e, S_{t_k-}=s\}}.
$$
and finally the relative values 
\begin{equation}
\label{eq:tildeI}
\tilde I_{s, e}(e') = \cfrac{\bar I_{s, e}(e')}{sup_{e'}\bar I_{s, e}(e')}
\end{equation}

Figure \ref{fig:I} displays this quantity for AXA, for each value of $s$, as an image (for which the vertical axis is $e$ and the horizontal axis is $e'$)\footnote{Let us note that, since for CAC40, most of the time the spread is of size $S=1$ (the only other possible state is $S=2$ which happens very rarely), this type of analysis is not relevant}. Readers are invited to refer to Appendix \ref{more_num_res} Fig.\ref{fig:I} for additional figures of the remaining stocks.
These figures show what seems to be a universal feature for stocks : the occurrence of a jump of size $e$ is essentially triggered by past occurrences of contrariant jump of size $-e$.
\begin{figure}
    \centering
    \includegraphics[width=0.8\textwidth]{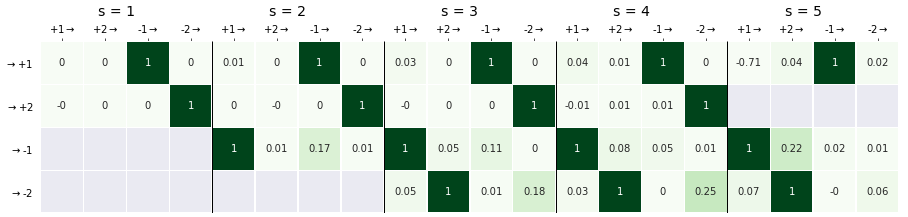}
    \caption{Relative kernel integrated quantities $\tilde I_{s, e}(e')$ for AXA. For each stock, for each value of $s$, an image is displayed showing $\tilde I_{s, e}(e')$ (defined by \eqref{eq:tildeI}) as a function of $e$ (vertical axis) and $e'$ (horizontal axis).}
    \label{fig:axa_I}
\end{figure}

\vskip .3cm
{\bf Comparison of the kernel shapes}

Finally, estimation results show that the most energetic kernels are decreasing "slowly", i.e., as a power-law $t^{-\beta}$ with an exponent $\beta \simeq 1$. Such a power-law shape of the cross-excitation kernels is not surprising since this behavior with similar exponent values have been observed by various authors when modelling the market activity \cite{BacryRev2015} or the dynamics of mid-price \cite{bacry2016estimation}.
Figure \ref{fig:axa_kern_contrariant} shows the contrariant kernels for the different stocks in a log-log plot. 
We see that they display a power-law behavior on 3 or 4 decades (depending on the kernel). (See Fig. \ref{fig:bnp_kern_contrariant}, \ref{fig:nokia_kern_contrariant}, \ref{fig:future_kern_contrariant} for the other assets in Appendix \ref{more_num_res}.)

Moreover, most of them display a very clear bump around the time $t \simeq 0.2ms$. This is not very surprising, this phenomenon has already been  revealed in several former works \cite{bacry2016estimation,rambaldi2017role}. It corresponds to an average value of the latency of the market itself, i.e., the average time for an agent to effectively place an order, reacting to a change of the order book. 
\begin{figure}[h]
    \centering
    \includegraphics[width=\textwidth]{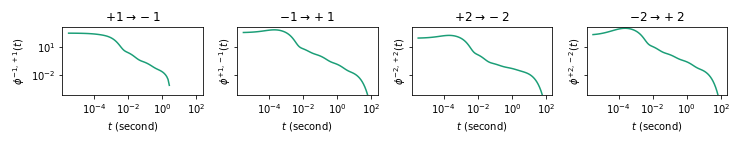}
    \caption{Hawkes kernels for AXA. In this figure we present four contrariant Hawkes kernels (i.e., with $e$ and $e'$ with different signs) . 
    Each kernel $\phi^{e,e'}$ (labeled $e' \rightarrow e$ on the figure with $e,e' \in {\cal E}$) represents the influence of the past jumps of size $e'$ on the occurrence probability of a future jumps of size $e$. Each kernel is represented by a sum of $L=6$ exponentials, i.e.,  $\phi^{e,e'}(t)$, where $\phi^{e,e'}(t)=\sum_l^L \alpha_l^{e,e'}\beta_l e^{-\beta_l t}$, where $\beta_l=\frac{1}{\tau_l}$ with $\tau_l$ taken in $\{10^{-4}s, 10^{-3}s, ..., 10^1s\}$
    All the kernels are displayed on a log-log scale and show a power law behavior on a large range of scales (3, 4 or even 5 decades)}
    \label{fig:axa_kern_contrariant}
\end{figure}
\subsection{Goodness-of-fit}
As we will show in this section, the model that we have built and estimated in the previous sections, is able to capture accurately very different statistical properties of the spread process.  
We will study successively,  the spread distribution itself (that has already been discussed in Section \ref{sec:hypersettings}), the inter-event-time distributions (i.e., time between change of spread values), the spread autocorrelation function and finally the auto-covariance function of the spread increment process.


\subsubsection{Spread Distributions}\label{spread_dist}
Let us first compare the spread distributions obtained on true data (both calendar time and event time) that have been already discussed in Section \ref{sec:hypersettings} and the ones obtained with data simulated by our model (fitted on true data). 

Following the same lines as in Section \ref{sec:hypersettings},
Fig.\ref{fig:axa_spread_dist} displays both calendar-time spread distribution, event-time distribution and spread jumps ($dS$) distribution for the AXA asset.  We see that the model fits very precisely all these distributions. 
Several works have explored the distribution of spread, but \cite{fosset2020endogenous} is one of the few that discuss it in detail. In their spread model (as shown in Equation \eqref{spread1}), the distribution of spread is geometric, given by : 
$$
\P(S\geq n) = \cfrac{1-\alpha_c}{1-\alpha}(1-r)r^{n-2}~~~~~~~~~~~\text{ for } n\geq 2
$$
where $\alpha_c = 1-\frac{\mu^+}{\mu^-}$ and $r\in (0,1)$ depends on $\alpha$ and $\beta$. However, since $\P(S=n)$ is a decreasing function for $n\geq 2$, the model Eq. \eqref{spread1} can only produce spread distributions that peak at $S=1$ or $S=2$. As shown in Figure \ref{fig:axa_spread_dist} and \ref{fig:spread_dist}, this is not consistent with actual data.

It is also important to note that the SDSH model is not limited to reproducing the spread distribution, but is a more comprehensive framework which can capture a wider range of spread dynamics.

\begin{figure}[h]
    \centering
    \includegraphics[width=\textwidth]{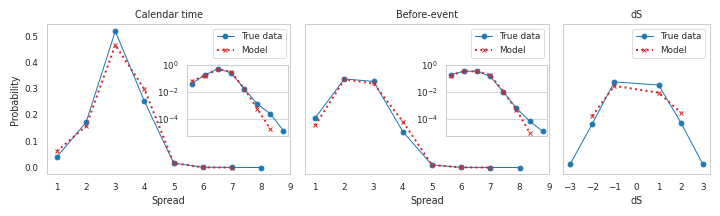}
    \caption{Spread distribution, comparison between true data and the the data obtained through simulation for AXA. The left-hand figure is the Calendar time distributions, the middle figure is the event time distributions, and the right-hand figure is the distributions of spread jumps size.}
    \label{fig:axa_spread_dist}
\end{figure}

\subsubsection{Inter-event time distributions}
Let $\{t_n\}_n$ be the successive times the spread process $S$ jumps (i.e., the spread changes). We define the  inter-event-times $\{\Delta t_n\}_n$  by 
$$
\Delta t_n = t_{n} - t_{n-1}.
$$
The first plot on the left of Fig.\ref{fig:axa_delta_t} displays the empirical  unconditional distribution of the inter-event-times for both the AXA true data and some simulated data using our model (fitted with AXA data). 
The two other plots on the right hand-side of this latter plot
show some conditional inter-event-time distributions. More precisely we define the set 
\begin{equation}
\label{eq:uncondinterev}
\{\Delta t_{S_1\to S_2}\} := \{\Delta t_i~\mid~ S(t_i+)=S_1,~ S(t_{i+1}+)=S_2\},
\end{equation}
where $S(t_i+)$ is the value of the spread immediately after the jump at time $t_i$.
The plot of Fig.\ref{fig:axa_delta_t} in the middle shows the distribution of $\{\Delta t_{1\to 3}\}$ whereas the one on the extreme right shows the distribution of $\{\Delta t_{3\to 2}\}$. 

 We see that the model performs extremely well in reproducing the unconditional and conditional inter-event-time distributions.
The plots in Fig.\ref{fig:axa_delta_t_qq} display the qq-plots of the true distributions versus the model distributions. Their  linear behavior show how good the fit of the model is. 

We invite the reader to look at the results obtained for BNP, NOKIA or CAC40 Future in Appendix \ref{app:interevent} (Figure \ref{fig:bnp_interevent_time}, \ref{fig:nokia_interevent_time} and \ref{fig:future_interevent_time}). 



\begin{figure}[h]
     \centering
    \begin{subfigure}[b]{0.8\textwidth}
         \centering
         \includegraphics[width=\textwidth]{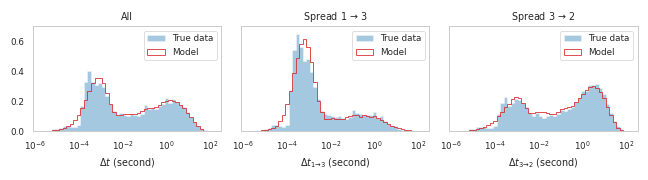}
         \caption{}
         \label{fig:axa_delta_t}
     \end{subfigure}
     \begin{subfigure}[b]{0.8\textwidth}
         \centering
         \includegraphics[width=\textwidth]{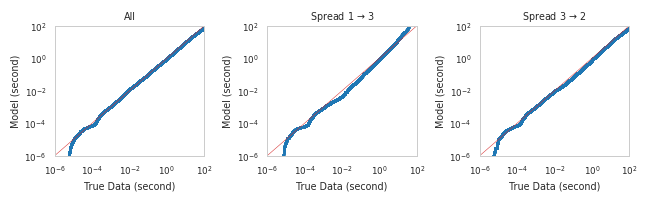}
         \caption{}
         \label{fig:axa_delta_t_qq}
     \end{subfigure}
     \caption{AXA Inter-event time unconditional and conditional distributions. Comparison between AXA true data and data obtained through simulation of our model (fitted on AXA data) (a) AXA true distributions versus model distributions for unconditional distribution or conditional distributions (see \eqref{eq:uncondinterev}. The x-axis is on log scale. (b) Corresponding qq-plots in  log-log scales.}
     \label{fig:axa_interevent_time}
\end{figure}


\subsubsection{Spread Autocorrelation}
In this section, we study the auto-correlation function of the spread for true data and check in which respect the SDSH spread is able to faithfully reproduce it. 

Estimation  of the autocorrelation function of the spread on true data could a priori be done using straightforward  quadratic covariations 
on all the available data. However such an estimation is likely to lead to some highly biased results. Indeed, it is well known that the orderbook dynamics (and consequently the spread dynamics) is subject to long range correlations and some very strong intraday seasonal effects \cite{bouchaud2009markets,gross2013predicting, fall2021forecasting}. The seasonal effects must be taken care of in come ways in order to avoid these biases. One classical to do so is to limit the computation of the quadratic covariations 
so some limited time slot every day (assuming that the seasonal effects between the different days of the weeks are of second order). 
In our case, we used 15min time slots on the real data/ More precisely we used everyday the eight 15min slots between 10am and 12. As far as the model simulated data are concerned we used a single 2 hour time-slot (since there is no seasonality in the model, the length of the time slot is not important). 

The plot of Fig.\ref{fig:covaxa} show these estimations for AXA (Additional plots for BNP, NOKIA and CAC40 are available in Appendix \ref{more_num_res} Fig.\ref{fig:spread-cor}). On each plot the autocorrelation function for the true data and the autocorrelation function for the simulated data are shown, both in lin-lin scale and in log-log scale. 
Again the fits are amazingly good. The model succeeds in reproducing the autocorrelation function of the spread with a very good accuracy. 

\begin{figure}[ht]
    \centering
    \includegraphics[width=0.5\textwidth]{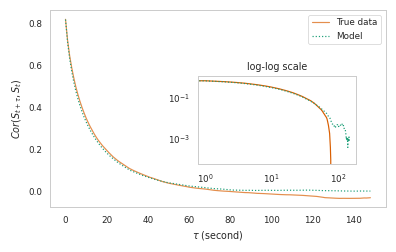}
     \caption{Auto-correlation function of the spread for AXA. True data versus model-simulated data. Each plot corresponds to a different asset. Quadratic variations were used for estimation. Eight 15min time-slots (between 10am and 12pm) were used everyday for true data in order to avoid intraday seasonal effects. 2 hours slots were used for model-simulated data.}
     \label{fig:covaxa}
\end{figure}

\subsubsection{Autocovariance of spread increments}
\label{sec:autocorrelation_spread}
In this section we focus on the autocovariance function of the spread increments. Let us first define the exact quantity under study.

Let us define the infinitesimal covariance of the infinitesimal measure $dS_t$, namely $Cov(dS_t,dS_{t'})$ which, assuming it is stationary, only depends on $(t'-t)$. Let us refer to it as $g(t'-t)$, i.e., 
\begin{eqnarray*}
g(t'-t)dtdt' & = & Cov(dS_t,dS_{t'}) \\
& = & 
\mathbf{E}[dS_t,dS_{t'}] - 
\mathbf{E}[dS_t] \mathbf{E}[dS_{t'}],
\end{eqnarray*}
Then, one gets, $\forall \delta >0$ and  $\forall \tau >0$,
\begin{eqnarray*} 
Cov(S_{t+\delta}-S_{t},S_{t+\delta+\tau}-S_{t+\tau}) 
& = & 
\mathbf{E} [\int_t^{t+\delta} dS_u \int_{t+\tau}^{t+\tau+\delta} dS_v] - 
\mathbf{E} [\int_t^{t+\delta} dS_u]
\mathbf{E}[\int_{t+\tau}^{t+\tau+\delta} dS_v] \\
& = & 
 \int_t^{t+\delta} \int_{t+\tau}^{t+\tau+\delta} \left(\mathbf{E}[dS_u dS_v] -  \mathbf{E}[dS_u] \mathbf{E}[dS_v]  \right)  \\
 & = & 
 \int_t^{t+\delta} \int_{t+\tau}^{t+\tau+\delta} g(v-u) dudv \\
 & = &
 \delta^2\int_0^{1} \int_{0}^{1} g(\tau+\delta(v-u)) dudv
\end{eqnarray*}
It is thus natural to introduce the normalized quantity (the relative covariance of spread increment during $\delta$ seconds with lag $\tau$ seconds)
\begin{equation}
\label{eq:ACV}
    ACV(\delta, \tau) := \frac{1}{\delta^2}Cov(S_{t+\delta}-S_{t},S_{t+\delta+\tau}-S_{t+\tau})
\end{equation}
Let us point out that the function $g(x)$ corresponds to the infinitesimal covariance function of a stationary process, it should thus decreases to 0 when the lag $x$ goes to infinity. It seems reasonable to assume that it does so in a "regular way", i.e., that there exist $\epsilon > 0$ and $K>0$ such that 
$$
|g'(x)| < K x^{-\epsilon},~~~\forall x>0
$$
Considering this assumption to be true, 
one gets, for any $\tau > \delta > \delta' > 0 $, that 
\begin{eqnarray*}
|ACV(\delta, \tau)-ACV(\delta', \tau)|  & = &  
\left| \int_0^{1} \int_{0}^{1} (g(\tau+\delta(v-u))- 
g(\tau+\delta'(v-u)))
dudv \right| \\
& \leq & (\delta-\delta') \max_{|y| \in [0,\max(\delta,\delta')]}{|g'(\tau + y)|}
 \int_0^{1} \int_{0}^{1} (v-u) \\
& \leq & (\delta-\delta') \max_{|y| \in [0,\max(\delta,\delta')]}{|g'(\tau + y)|}
\\
& \leq & K\frac {(\delta-\delta')}{(\tau-\delta)^\epsilon}
\end{eqnarray*}
Consequently, the covariance function $ACV(\delta,\tau)$ can be estimated independently of the value of $\delta$ as long as the lag $\tau$ is large enough compared to the value of $\delta$ that is being used. 
So in order to avoid having to perform estimation for all values of $\delta$ and $\tau$,  theoretically,  one could 
 fix a very small value for $\delta$ and then estimate the ACV function on a range of $\tau$ that statisfies $\tau >> \delta$. However, as we will see, this will lead to very (high frequency) noisy estimations. In order to smooth out the estimation one needs to use a value for $\delta$ that is much smaller that $\tau$ but not "too" small. 

Fig.\ref{fig:axa_acv_colors} not only illustrates this discussion but illustrates the fact that the model perfectly reproduces all the statistical features of the true data (of the AXA asset).
The insert in the figure shows for a fixed $delta=0.1s$ the function $-ACV(\delta,\tau)$ as a function of $\tau$ in a log-log scale. We chose to represent $-ACV$ instead of $ACV$ since, not surprisingly, due to the mean reversion property of the spread, we expect $ACV$ to be mainly negative. One sees that if $\delta$ is too small compared to $\tau$ the result gets extremely noisy. Moreover one sees that  the autocovariance function computed using the model simulated data reproduces very well the behavior of the one computed using the true data. 

The main plot displays the different estimations of the $ACV(\delta,\tau)$ as a function of $\tau$ for different values of $\delta$ (indicated on the legend on the right hand-side of the plot). Following what we just said, we limited for each $\delta$
the estimation of $ACV(\delta,\tau)$ for values of $\tau$ on a range so that $\delta$ is small compared to $\tau$ but not too small. Doing so, we see that all the estimation curves fall "on top of the over", letting discover a smooth curve for the auto-covariance function on almost 7 decades of values for $\tau$. 
This curve is close to be linear indicating that the auto-covariance function is close to be power-law. One can see the slight latency bump around $\tau =200\mu s$ that was already pointed out in Section \ref{sec:hawkes_kernel} when discussing Fig. \ref{fig:axa_kern_contrariant} and \ref{fig:kern_contrariant}. But what is really impressive is the way the estimation computed using the model-simulated data fits the estimation on true data on the whole range of scales. The fit is very accurate. 

Last but not least, let us point out that one  gets the exact same results when considering the other assets. This is illustrated in Fig. \ref{fig:all_assets_acv}. Not only the model fits are extremely impressive but the auto-covariance curve look very much alike across the assets.  

This seems to be the signature of what could be a stylized fact of the spread empirical processes


\begin{figure}[h]
     \centering
    \includegraphics[width=\textwidth]{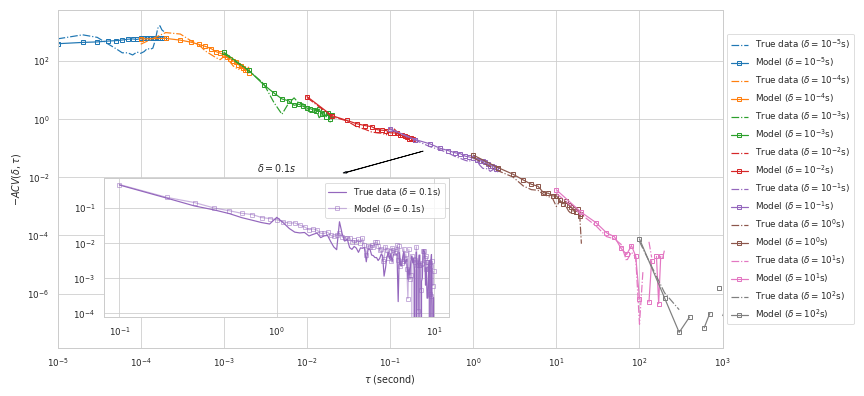}
    \caption{The $-ACV(\delta,\tau)$ functions for different values of $\delta$ as a function of $\tau$ using a log-log scale both for using the AXA true data and the model-simulated data (fitted on AXA true data). As expected (see discussion) the curves for different $\delta$'s fall one onto the over letting discover a power-law behavior with a slight bump close to the latency time scale $\simeq 200\mu s$ (see Section \ref{sec:hawkes_kernel}). The figure in the inset shows the particular curve $ACV(\delta,\tau)$ for $\delta=0.1s$. It clearly confirms the fact that if $\tau$ is too far to $\delta$ the estimation gets very noisy.}
    \label{fig:axa_acv_colors} 
\end{figure}

\section{Illustration on using SDSH model for spread forecasting}\label{sec:pred}
The goal of this section is to show that the SDSH model is a good candidate to be  used for high-frequency spread forecasting purposes. Spread forecasting is a difficult task that  deserves a full paper dedicated on it with comparisons with state of the art methods. It is clearly out of the scope of our paper. 
In this section, we are far less ambitious and we just intend to give some very preliminary results that allow us to glimpse the possibilities of such an application. 

The issue we consider is the prediction the spread value at a specific time horizon $\Delta$ in the very near future (i.e., $\Delta < 1$ min).
According to Fig. \ref{fig:axa_acv_colors}, a one-minute time window should be sufficient to capture the direct impact of past events on future spread values. Taking into account only the direct impact these events should be enough to capture most of the dynamics and to give a good idea of the order of the performance SDSH can reach. 
Since we don't have any explicit expression for the spread distribution, we rely on Monte Carlo simulations to estimate the expected spread value at the given time horizon. Let us be more precise. 

At time $t_0$,  the spread point process during the time interval $[t_0-60, t_0]$ is sum up by $\{S_u\}_{t_0-60<u<t_0}$ or alternatively by $S_{t_0-60}$ and all the events 
$\{(t_i, e_i), t_0-60\le t_i<t_0\}_{i=1,...,n}$ that occur within the time-interval $[t_0-60,t_0]$. 
Thus, in order to estimate $\mathbb{E}[S_{t_0+\Delta}|S_{u, u \in [t_0-60, t_0]}]$, i.e., the conditional mean at time-horizon $\Delta$, 
we simulate 100 processes from $t_0$ to $t_0+\Delta$ with the intensity function at time $t\in [t_0,t_0+\Delta]$ defined as
\begin{equation}
    \label{eq:lambda_pred}
\lambda^e_t = f^e(S_{t-})\bigprt{\tilde\mu^e(t) + \sum_{e'\in\mathcal{E}}\int_{t_0}^t \phi^{e,e'}(t-u)dS^{e'}_u} 
\end{equation}
where the baseline $\tilde\mu^e(t)$ is a function on $t\in[t_0, t_0+\Delta]$ defined by 
\begin{equation}
    \label{eq:mu_pred}
\tilde\mu^e(t)=\mu^e + \sum_{e'\in\mathcal{E}}\sum_{i=1}^n 1_{e_i=e'}\phi^{e,e'}(t-t_i) \; .
\end{equation}
Before illustrating the performance of such an estimator, let us briefly describe an alternative model, previously introduced in \cite{gross2013predicting}, that we will use as a benchmark for our numerical tests.

\paragraph{Introducing a benchmark: the Autoregressive Conditional Double Poisson (ACDP) model}
In \cite{gross2013predicting}, the authors introduce various spread prediction models which are all based on the use of autoregressive conditional Poisson model. As our focus is on high-frequency prediction, we will consider only the so-called ACDP model and not its long-memory version (i.e., the LMACP, the Long Memory Autoregressive Conditional Poisson model).
Again, let us point out that building an optimized SDSH-based forecast procedure and comparing it with the main state of the art procedures is out of the scope of this paper. A forthcoming work will be specifically dedicated to these issues.

Within the ACDP framework, the spread process $S_k, k\in \mathbb{Z}$ is a discrete time series, consisting of spread values subsampled every $\Delta=30s$ (as shown below, we explore the effects of varying this sub-sampling interval which plays also the role of the time-horizon between 3 seconds and 30 seconds) : 
\begin{equation}
    \begin{split}
        &\lambda_k = c + \alpha S'_{k-1} + \beta \lambda_{k-1}\\
        &S'_k|\mathcal{F}_{k-1} \sim \mathcal{DP}(\lambda_k, \gamma)
    \end{split}
\end{equation}
where $S'_k = S_k-1 \in \mathbb{N}$, and $\mathcal{DP}(\lambda_k, \gamma)$ stands for the Double Poisson distribution defined by
$$
\mathbb{P}(S'_k=n|\lambda_k, \gamma) = c(\gamma,\lambda_k)\gamma^{1/2}e^{-\gamma\lambda_k}\bigprt{\frac{e^{-n}n^n}{n!}}\bigprt{\frac{e\lambda_k}{n}}^{\gamma n}
$$
Under this model, $\mathbb{E}[S'_k|\mathcal{F}_{k-1}]=\lambda_k$.

As shown in \cite{gross2013predicting}, the log likelihood function of ACDP model reads:
$$
\log \mathcal{L}(c, \alpha, \beta, \gamma|S'_{[1:N]}) = \sum_{k=1}^T \Bigprt{\frac{1}{2}\log(\gamma)-\gamma\lambda_k + S'_k(\log S'_k-1)) - \log(S'_k!)+\gamma S'_k \bigprt{1+\log(\frac{\lambda_k}{S'_k})}}
$$
where $\lambda_t = (1-\beta B)^{-1}(c + \alpha B(S'_t))$ with $B(S'_t)=S'_{t-1}$ being the backshift operator which gives,
when $\beta<1$:
$$
\lambda_t = \sum_{i=0}^\infty \beta^i B^i (c + \alpha B(S'_t))
$$
In practice, one estimates $\lambda_t$ based on a truncation of the infinite sum, \ie, $\lambda_t = \sum_{i=0}^N \beta^i B^i (c + \alpha B(S'_t))$, where $N$ is an hyper parameter of the method. In this work, we set $N$ to be 60. In fact, after experimenting with different values of $N$, we found that the results remain consistent as long as $N$ is greater than 10. 

\paragraph{Numerical results}
Let us present the numerical results on the relative performances of SDSH and ACDP based methods. We are going to compare three different predictors for time-horizons varying from $\Delta=3s$ up to $\Delta=30s$. 
\begin{itemize}
    \item[-] {\bf SDSH predictor} :
    \begin{itemize} 
    \item[-] For each stock, the SDSH model is calibrated on training data which consists of 30-day spread dynamics within a 2-hour window each day (from 10am to 12pm) for avoiding strong seasonal effects. 
    \item[-] No re-estimation of the model parameters is made along running the test
    \item[-] At each time, forecasting is made using the previous 1 minute data following Eqs \eqref{eq:lambda_pred} and \eqref{eq:mu_pred}
    \end{itemize}
\item {\bf ACDP predictor} :
\begin{itemize}
    \item[-] For each stock, the ACDP model is calibrated using a time rolling window of 1-hour of spread data subsampled every $\Delta$ seconds.  
    \item[-] Parameters are  reestimated every 10 minutes to ensure that the predictions remain up-to-date. 
    \item[-] The model employs a one-step ahead prediction strategy using the most recent estimated model for
forecasting
\end{itemize}
\item {\bf Last predictor} : A simple benchmark predictor that consists in considering a that $S_t$ is a martingale, i.e., in taking the last observable value of the spread $\hat{S}_{t_0+\Delta} = S_{t_0}$ as the best prediction. 
\end{itemize}
All methods are evaluated over a test period lasting 50 consecutive days, following the initial 30-day training period for the SDSH model. The evaluation takes place specifically from 11am to 12pm every day. This time-frame is chosen to accommodate the ACDP method, as it requires a minimum of one hour to calibrate the model before generating predictions. Therefore, we cannot make predictions before 11am. 

We evaluate these two predictors on three stocks (AXA, BNP Paribas and Nokia) as well as the Cac40 index future. The performance comparison is presented in Table \ref{tab:hyperparam}. 

\begin{table}[!h]
\centering
\begin{subtable}[h]{0.49\textwidth}
\setlength\tabcolsep{7pt}
\centering
\begin{tabular}{ccccc}
\hline
$\Delta$ & 3s & 6s & 12s & 30s\\
\hline
{Last} & 0.408 & 0.587 & 0.784 & 1.082\\
ACDP & 0.514 & 0.520 & 0.565 & 0.808\\
SDSH & \textbf{0.363} & \textbf{0.478} & \textbf{0.561} & \textbf{0.648}\\
\hline
\end{tabular}
\caption{AXA}
\end{subtable}
\begin{subtable}[h]{0.49\textwidth}
\centering
\setlength\tabcolsep{6pt}
\begin{tabular}{ccccc}
\hline
$\Delta$ & 3s & 6s & 12s & 30s\\
\hline
Last & 0.833 & 1.177 & 1.487 & 1.818\\
ACDP & 0.737 & \textbf{0.851} & 0.971 & 1.320\\
SDSH & \textbf{0.692} & 0.865 & \textbf{0.964} & \textbf{1.065} \\
\hline
\end{tabular}
\caption{BNP}
\end{subtable}
\begin{subtable}[h]{0.49\textwidth}
\centering
\setlength\tabcolsep{7pt}
\begin{tabular}{ccccc}
\hline
$\Delta$ & 3s & 6s & 12s & 30s\\
\hline
Last & 0.388 & 0.609 & 0.904 & 1.340\\
ACDP & 0.508 & 0.609 & \textbf{0.744} & \textbf{1.031} \\
SDSH & \textbf{0.361} & \textbf{0.543} & 0.766 & 1.081\\
\hline
\end{tabular}
\caption{NOKIA}
\end{subtable}
\begin{subtable}[h]{0.49\textwidth}
\centering
\setlength\tabcolsep{7pt}
\begin{tabular}{ccccc}
\hline
$\Delta$ & 3s & 6s & 12s & 30s\\
\hline
Last & 0.370 & 0.421 & 0.445 & 0.457\\
ACDP & 0.283 & 0.258 & 0.256 & 0.267\\
SDSH & \textbf{0.230} & \textbf{0.240} & \textbf{0.244} & \textbf{0.245}\\
\hline
\end{tabular}
\caption{Cac40 Index Future}
\end{subtable}
\caption{Mean Square Error (MSE) of the 3 different predictors, Last, ACDP and SDSH for different time sub-samplings $\Delta$ (which plays also the role of the time-horizon).}
\label{tab:hyperparam}
\end{table}

As table \ref{tab:hyperparam} shows, the Last predictor is always worse than the SDSH predictor though it outperforms the ACDP predictor sometimes at the highest frequency (i.e., $\Delta=3s$, for AXA and NOKIA).
Moreover the SDSH predictor outperforms most of the time the ACDP predictor (and systematically for the highest frequency $\Delta=3s$). These results are very encouraging as far as the SDSH prediction performances are concerned. A Detailed comparison with several state of the art predictors and for a wider range of time horizons will be addressed in a forthcoming work.

\section{Conclusion}\label{sec:conclusion}
In this paper, we introduced a State Dependent Hawkes process (SDSH) for modelling bid-ask spread fluctuations, which generalizes the spread models presented in \cite{zheng2014modelling} and \cite{fosset2020endogenous}. Our model is a 2K-variate Hawkes process which can accommodate different jump sizes (ranging from 1 to K). In order to account for the current spread value $S_t$, we introduced a spread-dependent term $f^e(S_t)$ (where $e$ denotes an event type) and multiplied the classical Hawkes intensity by this term. We chose to use the sum of exponential kernels to benefit from the Markovian properties of the model. Notably, we demonstrated the ergodicity property in a particular case, indicating that the spread process converges to a stationary distribution over time.

Then we calibrated our SDSH model using high-frequency data obtained from Cac40 Euronext Market, including three stocks (AXA, BNP, Nokia) and the Cac40 Future index. We examined the estimated spread-dependent term $f$, as well as kernel functions to better understand how they affect the spread dynamics. Our analysis revealed that the estimated $f^e(\cdot)$ is decreasing when $e$ is an event which increased the spread, while for at least one downward event $e'$, $f^{e'}(\cdot)$ is increasing. Such $f^e$ is able to press the spread down when its value is high by exciting more downward events. In terms of kernel functions, our estimation results showed that most of them tend to decrease very slowly, resembling a power-law kernel function. 

We studied the ability of this model to capture various spread statistics. We found that our model very successfully replicates the spread distributions, measured by both calendar time and event time. We also observed that the model accurately reproduces other important statistical features, including the distributions of inter-event times, spread autocorrelations as well as spread increments autocovariance. 

Finally we demonstrated the effectiveness of the SDSH model in predicting the spread across different sample window sizes. Our results suggest that the SDSH model is a reliable and robust choice for predicting the spread. 

As part of future work, we plan to expand our spread model to higher dimension, for example by incorporating bid and ask processes. 

\section*{Acknowledgements}
We thank Euronext for making their data available to us. This research is partially supported by the Agence Nationale de la Recherche as part of the “Investissements d'avenir” program (reference ANR-19-P3IA-0001; PRAIRIE 3IA Institute).

\section*{Appendices}
\begin{appendices}
\section{Proof of V-uniform ergodicity}
\label{annex:ergo}
We restrict our model to the case where $K=1$ and $L=1$. The model then writes 
$$\lambda^e(t) = f^{e}(S_{t-})(\mu^e + \sum_{e'\in\mathcal{E}} \int_0^t \phi^{e, e'}(t-s) dS^{e'}_s) $$
where $\mathcal{E}=\{+1,-1\}=:\{+,-\}$ and  $\phi^{e, e'}(t)=\alpha^{e,e'}\beta e^{-\beta t}$

This section is devoted to the prove of the Proposition \ref{thm_ergodicity} of Section \ref{sec:markov} that we replicate here : 
\begin{proposition}
Let us consider the model given just above, assume that the following conditions are satisfied
\begin{empheq}[left={(A)=}\empheqlbrace]{align}
	&f^-(1) = 0 \tag{{A\textsubscript{1}}}\\
	&f^{-}(S) \geq \gamma S \text{ for some $\gamma > 0$ when $S\geq2$} \tag{{A\textsubscript{2}}}\\
	&\sup_{S}\{f^+(S)\}(\alpha^{+,-}+\alpha^{+,+}) < 1 \tag{{A\textsubscript{3}}}
\end{empheq}	
then the process $(S_t,X_t)$ is a V-uniformly ergodic Markov process.
\end{proposition}
\begin{remark}
    It is worth noting that in this model, the last condition $\sup_{S}\{f^+(S)\}(\alpha^{+,-}+\alpha^{+,+}) < 1$ is equivalent to $(\alpha^{+,-}+\alpha^{+,+}) < 1$ and $\sup_{S}\{f^+(S)\} <= 1$. 
    
\end{remark}
\begin{remark}
	From now, we replace event $+$ by $1$ and $-$ by $2$. To avoid the ambiguity of notation, in the following part, we replace the superscripts by subscripts. 
\end{remark}
\begin{proposition}
Define $X_{e,e'}(t) := \int_0^t  \phi^{e,e'}(t-s) dN_s^{e'}$. Then $(X, S)$ is Markovian, where $X=(X_{lm})_{l,m\in\{1,2\}}$.
\end{proposition}

\subsection*{Lyapunov function}
%

As $X_{lm}(t) = \int_0^t  \alpha_{lm} e^{-\beta(t-s)}dN_s^m$, 
$$ dX_{lm}(t) = -\beta X_{lm}(t)dt + \alpha_{lm}dN^m_t$$

\paragraph{Function $V_{lm}(X_{lm})$} 
$$
V_{lm}(X_{lm}) = X_{lm}
$$
The infinitesimal generator $\mathcal{L}$ of $X_{lm}$ on $V_{lm}$
\begin{eqnarray*}
	\mathcal{L}V_{lm}(X)&=& \alpha_{lm} \lambda_m - \beta X_{lm}\\
					&=&  -\beta  X_{lm} + \alpha_{lm} f_m(S) \mu_m +  \alpha_{lm}  f_m(S) (\sum_n \beta X_{mn})\\
					&=& -\beta X_{lm} + \alpha_{lm}\mu_m f_m(S) + \alpha_{lm}\beta f_m(S) (X_{m1}+X_{m2})\\
\end{eqnarray*}

\paragraph{Function $V_S(S)$}
$$
V_2(S) = S
$$
As $dS_t = dN^{1}_t - dN^{2}_t$, the infinitesimal generator $\mathcal{L}$ of $S$ on $V_S$:
\begin{eqnarray*}
	\mathcal{L}V_2(S) &=& \lambda_{1} - \lambda_{2}\\
					  &=& f_{1}(S)\mu_{1}+\beta f_{1}(S)X_{11}+\beta f_{1}(S)X_{12}-f_{2}(S)\mu_{2}- \beta f_{2}(S)X_{21}- \beta f_{2}(S)X_{22}
\end{eqnarray*}

\paragraph{Function V on $(X,S)$}
Now we consider a function $V$ on $(X, S)$
$$
V(X, S) = \sum_{l,m\in\{1,2\}}\eta_{lm}X_{lm} + \eta S
$$
where $\eta_{lm}, \eta > 0$.

Then the infinitesimal generator $\mathcal{L}$ of $(X,S)$ on $V$ is:
\begin{eqnarray*}
	\mathcal{L}V(X,S)
		&=& -\eta_{11}\beta X_{11} + \eta_{11}\alpha_{11}\mu_1f_1(S) + \eta_{11}\alpha_{11}\beta f_1(S) (X_{11}+X_{12})\\
		& & -\eta_{12}\beta X_{12} + \eta_{12}\alpha_{12}\mu_2f_2(S) + \eta_{12}\alpha_{12}\beta f_2(S) (X_{21}+X_{22})\\
		& & - \eta_{21}\beta X_{21} + \eta_{21}\alpha_{21}\mu_1f_1(S) + \eta_{21}\alpha_{21}\beta f_1(S) (X_{11}+X_{12})\\
		& & -\eta_{22}\beta X_{22} + \eta_{22}\alpha_{22}\mu_2f_2(S) + \eta_{22}\alpha_{22}\beta f_2(S) (X_{21}+X_{22})\\
		& & + \eta f_{1}(S)\mu_{1}+ \eta\beta f_{1}(S)X_{11}+ \eta\beta f_{1}(S)X_{12}\\
		& & - \eta f_{2}(S)\mu_{2}- \eta\beta f_{2}(S)X_{21}- \eta\beta f_{2}(S)X_{22}
\end{eqnarray*}

Before giving the proof of Theorem \ref{thm_ergodicity}, we should mention the following theorem. See Theorem 5.2 in \cite{down1995exponential} and 2.5.2 in \cite{abergel2013mathematical}. 
\begin{theorem}
For a $\psi$-irreducible, aperiodic Markov process $X$, if the following drift condition ($\mathcal{D}$) holds, then $X$ is V-uniformly ergodic.

\begin{center}
    ($\mathcal{D}$)  For some $\rho, b > 0$ and a coercive function $V \geq 1$ 
        \begin{equation}
			\mathcal{L}V\leq -\rho V + b
        \end{equation}      
\end{center}
\end{theorem}

\begin{proof}[Proof of Theorem \ref{thm_ergodicity}]\label{proof}
Suppose that we have already the following conditions:
\begin{empheq}[left={(\mathcal{H})=}\empheqlbrace]{align}
	& \eta_{12}\alpha_{12}+\eta_{22}\alpha_{22} < \eta \tag{H\textsubscript{1}} \\
	& \eta_{11}\alpha_{11}+\eta_{21}\alpha_{21} + \eta < \eta_{11}\tag{H\textsubscript{2}}\\
	& \eta_{11}\alpha_{11}+\eta_{21}\alpha_{21} + \eta < \eta_{12}\tag{H\textsubscript{3}}
\end{empheq}

To proceed, we will split $\mathcal{L}V(X,S)$ into four distinct parts and show their individual upper bounds.
\begin{equation*}
    \mathcal{L}V(X,S) = I_1 + I_2 + I_3 + I_4
\end{equation*}
where 
\begin{equation*}
\begin{split}
I_1 &= (1)+(2) \\
&= -\eta_{11}\beta X_{11} - \eta_{12}\beta X_{12} \\
&+ \eta_{11}\alpha_{11}\beta f_1(S) (X_{11}+X_{12})+\eta_{21}\alpha_{21}\beta f_1(S) (X_{11}+X_{12})+\eta\beta f_{1}(S)(X_{11}+X_{12})\\
&< (-\eta_{11}+\eta_{11}\alpha_{11}+\eta_{21}\alpha_{21}+\eta)\beta X_{11} + (- \eta_{12}+\eta_{11}\alpha_{11}+\eta_{21}\alpha_{21}+\eta)\beta X_{12}\\
&< -\epsilon_1\beta X_{11} - \epsilon_2\beta X_{12}
\end{split}
\end{equation*}
this last inequality is directly derived by conditions $(H_2)$ and $(H_3)$. 
\begin{equation*}
\begin{split}
I_2 &= (3)+(4) \\
&= -\eta_{21}\beta X_{21} -\eta_{22}\beta X_{22} \\
&+ \eta_{12}\alpha_{12}\beta f_2(S) (X_{21}+X_{22})+\eta_{22}\alpha_{22}\beta f_2(S) (X_{21}+X_{22})-\eta\beta f_{2}(S)(X_{21}+X_{22})\\
&= -\eta_{21}\beta X_{21} -\eta_{22}\beta X_{22} + (\eta_{12}\alpha_{12}+\eta_{22}\alpha_{22}-\eta)\beta f_2(S) (X_{21}+X_{22})\\
&\overset{(H_1)}{<} -\eta_{21}\beta X_{21} -\eta_{22}\beta X_{22}
\end{split}
\end{equation*}
Now for the other two terms $I_3$ and $I_4$:
\begin{equation*}
	\begin{split}
		I_3 &= (5a)\\
		& = \eta_{11}\alpha_{11}\mu_1f_1(S) + \eta_{21}\alpha_{21}\mu_1f_1(S) + \eta \mu_{1}f_{1}(S)
		< \eta_{11}\alpha_{11}\mu_1 + \eta_{21}\alpha_{21}\mu_1 + \eta\mu_{1} =: C_1
		\end{split}
\end{equation*}
\begin{equation*}
	\begin{split}
		I_4 &= (5b)\\
		&= \eta_{12}\alpha_{12}\mu_2f_2(S)+\eta_{22}\alpha_{22}\mu_2f_2(S)-\eta f_{2}(S)\mu_{2} = (\eta_{12}\alpha_{12}+\eta_{22}\alpha_{22}-\eta)f_{2}(S)\mu_{2}\\
		&< -\epsilon_0\mu_2 f_2(S) \overset{(A_2)}{<} \epsilon_0\mu_2\gamma-\epsilon_0\mu_2\gamma S =: C_2 - -\epsilon_0\mu_2\gamma S
		\end{split}
\end{equation*}
where 
\begin{itemize}
	\item[-]$0< \epsilon_0 < \eta-\eta_{12}\alpha_{12}-\eta_{22}\alpha_{22}$
	\item[-] $0< \epsilon_1 < \eta_{11} - (\eta_{11}\alpha_{11}+\eta_{21}\alpha_{21}+\eta)$
	\item[-]$0< \epsilon_2 < \eta_{12} - (\eta_{11}\alpha_{11}+\eta_{21}\alpha_{21}+\eta)$
\end{itemize}

Therefore
\begin{equation*}
    \begin{split}
    \mathcal{L}V(X,S) &= I_1 + I_2 + I_3 + I_4\\
	& <  -\epsilon_1\beta X_{11} - \epsilon_2\beta X_{12} - \eta_{21}\beta X_{21} -\eta_{22}\beta X_{22} +  C_1 - \epsilon_0\mu_2 f_2(S)\\
	& <  -\epsilon_1\beta X_{11} - \epsilon_2\beta X_{12} -\eta_{21}\beta X_{21} -\eta_{22}\beta X_{22}  + C_1 + C_2 - \epsilon_0\mu_2\gamma S \\
	& <  -\rho (\eta_{11}X_{11}+\eta_{12} X_{12}+\eta_{21} X_{21}+\eta_{22}X_{22}+\eta S) + C\\
	& =  -\rho V(X,S) + C
    \end{split}
\end{equation*}
where $\rho = \min\{\cfrac{\epsilon_1}{\eta_{11}}, \cfrac{\epsilon_2}{\eta_{12}}, 1, \cfrac{\epsilon_0\mu_2\gamma}{\eta}\}, \beta > 0$ and $C = C_1 + C_2$

Now we only need to find some $\eta_{lm}, \eta > 0$ satisfying the hypothesis $(\mathcal{H})$ to finish this proof.

As $\alpha_{11}+\alpha_{12} < 1$, $\cfrac{\alpha_{11}}{1-\alpha_{12}} <1< \cfrac{1-\alpha_{11}}{\alpha_{12}}  $. We note $\delta = \cfrac{1}{2}(\cfrac{1-\alpha_{11}}{\alpha_{12}}  - 1) $. Then the following values for $\eta_{lm}, \eta$
\begin{equation}\label{eta_0}
    \begin{cases}
	\eta_{11}=1,\eta_{12} = \cfrac{1-\alpha_{11}}{\alpha_{12}} - \delta > 1 = \eta^{11}\\
	\eta_{21} = \delta\cfrac{\alpha_{12}}{4\alpha_{21}}, \eta_{22} =\delta\cfrac{\alpha_{12}}{4\alpha_{22}}\\	
	\eta = 1-\alpha_{11} - \cfrac{1}{2}\delta\alpha_{12}
\end{cases}
\end{equation}
satisfy the condition $(\mathcal{H})$.
\end{proof}

\begin{remark}\label{ergo2}
For a more complicated version of our model :
$$\lambda^e(t) = f^{e}(S_{t-})(\mu^e + \sum_{e'\in\mathcal{E}} \int_0^t \phi^{e, e'}(t-s) dS^{e'}_s) $$
where $\mathcal{E}=\{+1,+2,-1,-2\}$, $\phi^{e, e'}(t)=\alpha^{e,e'}\beta e^{-\beta t}$ (exponential kernels). 
Using the same proof, we can prove that under the following conditions $(\mathcal{A}_1)$ and $(\mathcal{H}_1)$, the $(X,S)$ is a V-uniformly ergodic Markov process.\\
$(\mathcal{A}_1) = 
\begin{cases}
    f^{-1}(S)=0 \mbox{ when } S=1\\
    f^{-2}(S)=0 \mbox{ when } S=1,2\\
	\max(f^{-1}(S), f^{-2}(S)) \geq \gamma S \mbox{ for some } \gamma > 0 \mbox{ when } S \geq 3\\
	f^{+1}(S), f^{+2}(S) \leq 1 \mbox{ for all } S
\end{cases}
$

$(\mathcal{H}_1) = 
\begin{cases}
		\sum_e\eta_{e,-1}\alpha^{e,-1} < \eta\\
		\sum_e\eta_{e,-2}\alpha^{e,-2} < 2\eta\\
		\sum_e\eta_{e,+1}\alpha^{e,+1} + \eta < \eta_{+1,e'}, \forall e'\in \mathcal{E}\\
		\sum_e\eta_{e,+2}\alpha^{e,+2} + 2\eta < \eta_{+2,e'}, \forall e'\in \mathcal{E}
\end{cases}$\\

And the coercive function $V$ is 
$
V(X, S) = \sum_{e,e'\in\mathcal{E}}\eta_{e,e'}X_{e,e'} + \eta S
$
\end{remark}

\section{Log-likelihood function of spread model}\label{mle}
The log-likelihood function for the spread model is a function on $\mu$, $\alpha$ and $f$.
For brevity, we will only give the formula for the case where $L=1$ (only one decay). In order to distinguish from the above notations, the log-likelihood function is denoted by $\mathbb{L}$.
\begin{equation}
\begin{split}
\mathbb{L}(\alpha,\mu, f) &= \sum_{e=1}^{2K} (-\int_0^T \lambda^e(t)dt + \int_0^T\log \lambda^e(t)dS^e_t)\\
  &= \sum_{e=1}^{2K} \sum_{k=1}^{S^e(T)}\log(\mu^e + \sum_{e'=1}^{2K} \alpha^{ee'}\beta\int_0^{t_k^e}e^{-\beta(t_k^e-s)}dS^j_s) + \sum_{e=1}^{2K} \sum_{k=1}^{S^e(T)} \log {f^e(S_{t^e_k})}\\
  & - \sum_{e=1}^{2K} \int_0^T(\mu^e + \sum_{e'=1}^{2K} \alpha^{ee'}\beta\int_0^{t}e^{-\beta(t-s)}dS^{e'}_s)f^e(S_t)dt
\end{split}
\end{equation}
where $S^e(T)$ is the number of event $e$ in $[0,T]$, and $t^e_k$ is the timestamp where the $kth$ event (type $e$) occurs.

\section{More numerical results}\label{more_num_res}

\begin{figure}[h]
    \centering
    \begin{subfigure}[b]{\textwidth}
    \centering
    \includegraphics[width=0.65\textwidth]{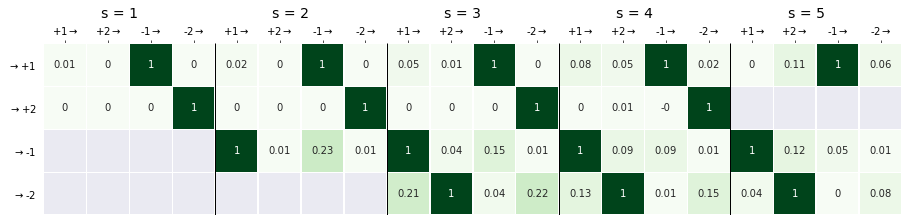}
    \caption{BNP}
    \end{subfigure}
    \begin{subfigure}[b]{\textwidth}
    \includegraphics[width=\textwidth]{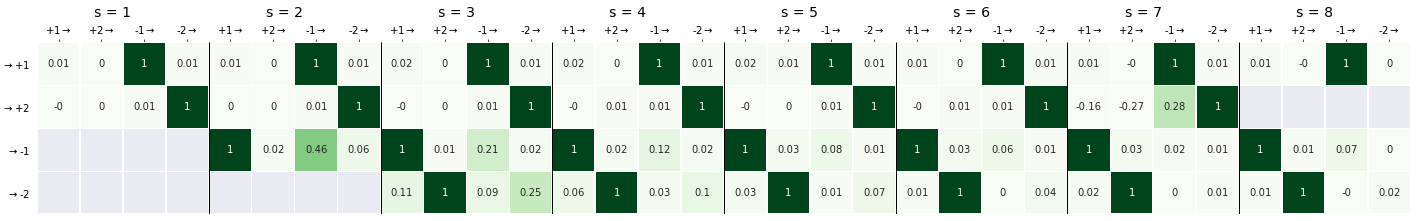}
    \caption{NOKIA}
    \end{subfigure}
    \caption{\textbf{Relative kernel integrated quantities} $\tilde I_{s, e}(e')$ for (a)BNP and (b) NOKIA. For each stock, for each value of $s$, an image is displayed showing $\tilde I_{s, e}(e')$ (defined by \eqref{eq:tildeI}) as a function of $e$ (vertical axis) and $e'$ (horizontal axis).}
    \label{fig:I}
\end{figure}

\begin{figure}[ht]
    \centering
    \begin{subfigure}[b]{\textwidth}
    \includegraphics[width=\textwidth]{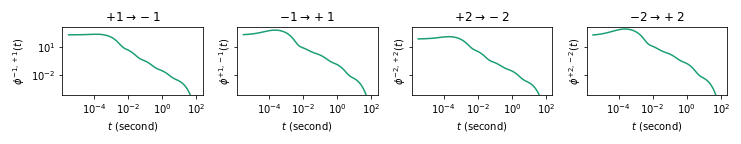}
    \caption{BNP}
    \label{fig:bnp_kern_contrariant}
    \end{subfigure}
    \begin{subfigure}[b]{\textwidth}
    \includegraphics[width=\textwidth]{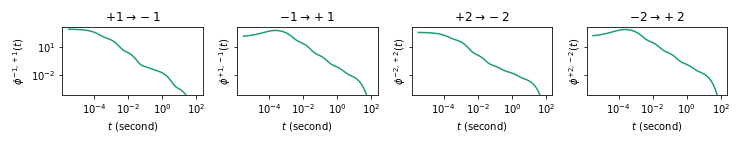}
    \caption{NOKIA}
    \label{fig:nokia_kern_contrariant}
    \end{subfigure}
    \begin{subfigure}[b]{0.52\textwidth}
    \includegraphics[width=\textwidth]{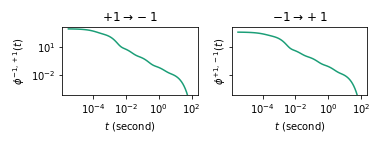}
    \caption{Cac40 Future}
    \label{fig:future_kern_contrariant}
    \end{subfigure}
    \caption{\textbf{Hawkes kernel shapes} for (a) BNP, (b) NOKIA, and (c) CAC40 Future. In this figure for each asset we present four contrariant Hawkes kernels (i.e., with $e$ and $e'$ with different signs (two contrariant Hawkes kernels for CAC40 Future) . 
    Each kernel $\phi^{e,e'}$ (labelled $e' \rightarrow e$ on the figure with $e,e' \in {\cal E}$) represents the influence of the past jumps of size $e'$ on the occurrence probability of a future jumps of size $e$. Each kernel is represented by a sum of $L=6$ exponentials, i.e.,  $\phi^{e,e'}(t)$, where $\phi^{e,e'}(t)=\sum_l^L \alpha_l^{e,e'}\beta_l e^{-\beta_l t}$, where $\beta_l=\frac{1}{\tau_l}$ with $\tau_l$ taken in $\{10^{-4}s, 10^{-3}s, ..., 10^1s\}$
    All the kernels are displayed on a log-log scale and show a power law behavior on a large range of scales (3, 4 or even 5 decades)}
    \label{fig:kern_contrariant}
\end{figure}

\begin{figure}
	\centering
	\begin{subfigure}[b]{\textwidth}
		\centering
		\includegraphics[width=0.9\textwidth]{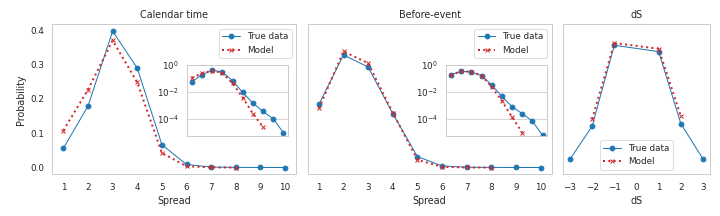}
		\caption{BNP}
		\label{fig:bnp_spread_dist}
	\end{subfigure}
	\begin{subfigure}[b]{\textwidth}
		\centering
		\includegraphics[width=0.9\textwidth]{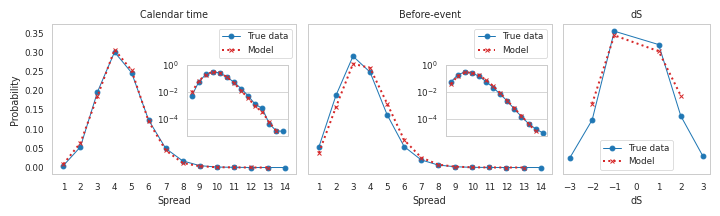}
		\caption{NOKIA}
		\label{fig:nokia_spread_dist}
	\end{subfigure}
	\begin{subfigure}[b]{\textwidth}
		\centering
		\includegraphics[width=0.9\textwidth]{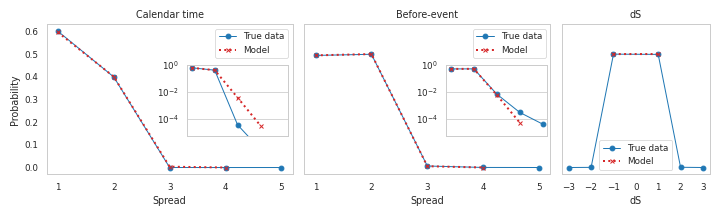}
		\caption{Cac40 Future}
		\label{fig:future_spread_dist}
	\end{subfigure}
	
	\caption{\textbf{Spread distributions}, comparison between true data and the the data obtained through simulation of (a) AXA, (b) BNP, (c) NOKIA, (d) CAC40 Future. The left-hand figures are the Calendar time distributions, the middle figures are the event time distributions, and the right-hand figures are the distributions of spread jumps size.}
	\label{fig:spread_dist}
\end{figure}

\label{app:interevent}
\begin{figure}[h]
     \centering
    \begin{subfigure}[b]{\textwidth}
         \centering
         \includegraphics[width=0.8\textwidth]{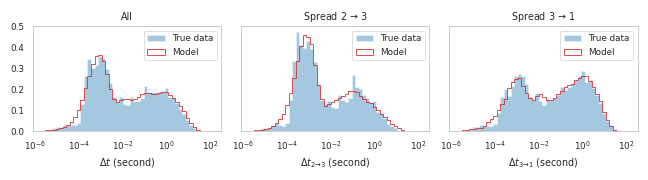}
         \caption{}
         \label{fig:bnp_delta_t}
     \end{subfigure}
     \begin{subfigure}[b]{\textwidth}
         \centering
         \includegraphics[width=0.8\textwidth]{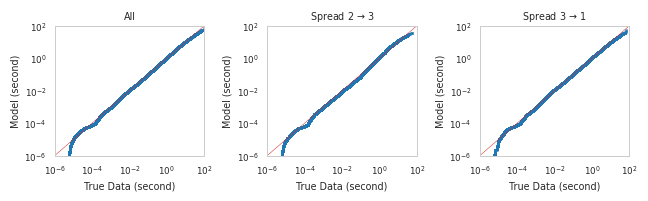}
         \caption{}
         \label{fig:bnp_delta_t_qq}
     \end{subfigure}
     \caption{\textbf{BNP Inter-event time distributions}. (a) Compare empirical inter-event-time distribution with simulated inter-event-time distribution. The x-axis is on log scale. (b) Log-log qq-plot of empirical inter-event-times (x-axis) \textit{vs.} simulated inter-event-times by model (y-axis).}
     \label{fig:bnp_interevent_time}
\end{figure}

\begin{figure}[h]
     \centering
    \begin{subfigure}[b]{\textwidth}
         \centering
         \includegraphics[width=0.8\textwidth]{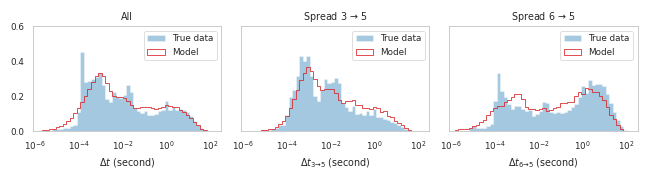}
         \caption{}
         \label{fig:nokia_delta_t}
     \end{subfigure}
     \begin{subfigure}[b]{\textwidth}
         \centering
         \includegraphics[width=0.8\textwidth]{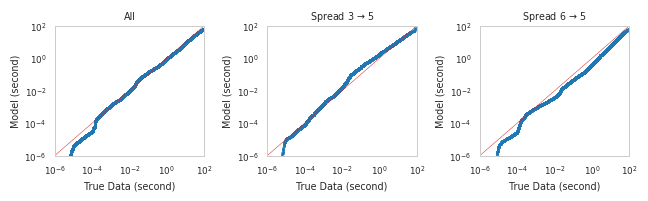}
         \caption{}
         \label{fig:nokia_delta_t_qq}
     \end{subfigure}
     \caption{\textbf{NOKIA Inter-event time distributions}. (a) Compare empirical inter-event-time distribution with simulated inter-event-time distribution. The x-axis is on log scale. (b) Log-log qq-plot of empirical inter-event-times (x-axis) \textit{vs.} simulated inter-event-times by model (y-axis).}
     \label{fig:nokia_interevent_time}
\end{figure}

\begin{figure}[h]
     \centering
    \begin{subfigure}[b]{\textwidth}
         \centering
         \includegraphics[width=0.8\textwidth]{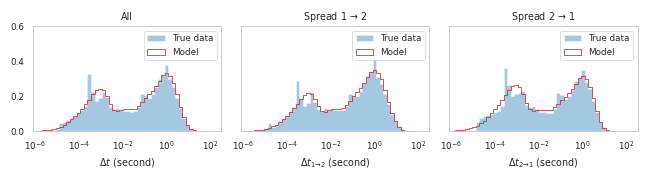}
         \caption{}
         \label{fig:future_delta_t}
     \end{subfigure}
     \begin{subfigure}[b]{\textwidth}
         \centering
         \includegraphics[width=0.8\textwidth]{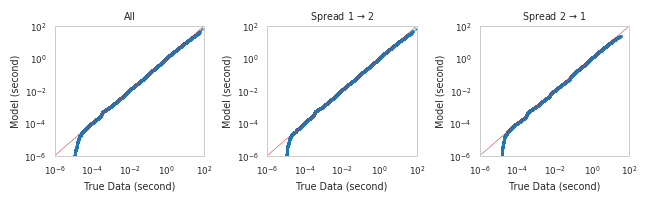}
         \caption{}
         \label{fig:future_delta_t_qq}
     \end{subfigure}
     \caption{\textbf{Cac40 index Future Inter-event time distributions} (a) Compare empirical inter-event-time distribution with simulated inter-event-time distribution. The x-axis is on log scale. (b) Log-log qq-plot of empirical inter-event-times (x-axis) \textit{vs.} simulated inter-event-times by model (y-axis).}
     \label{fig:future_interevent_time}
\end{figure}

\begin{figure}[h]
\centering
    \begin{subfigure}[b]{0.4\textwidth}
         \includegraphics[width=\textwidth]{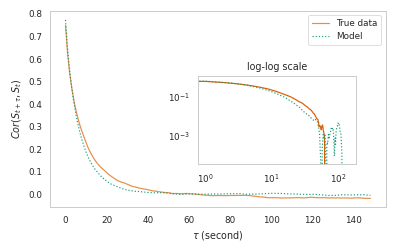}
         \caption{BNP}
         \label{fig:covbnp}
     \end{subfigure}
    \begin{subfigure}[b]{0.4\textwidth}
         \includegraphics[width=\textwidth]{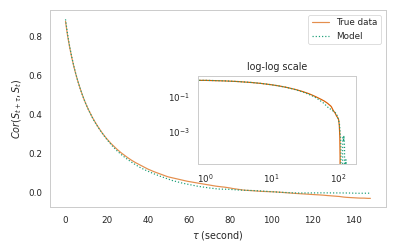}
         \caption{NOKIA}
         \label{fig:covnokia}
     \end{subfigure}
    \begin{subfigure}[b]{0.4\textwidth}
         \includegraphics[width=\textwidth]{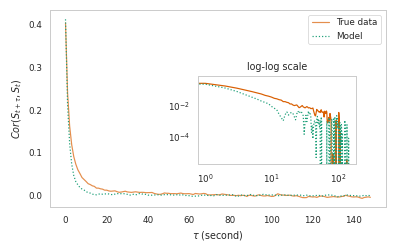}
         \caption{Cac40 Future}
         \label{fig:covfuture}
     \end{subfigure}
     \caption{\textbf{Spread autocorrelation} for (a) BNP, (b) Nokia, (c) Cac40 index Future. True data versus model-simulated data. Each plot corresponds to a different asset. Quadratic variations were used for estimation. Eight 15min time-slots (between 10am and 12pm) were used everyday for true data in order to avoid intraday seasonal effects. 2 hours slots were used for model-simulated data.}
     \label{fig:spread-cor}
\end{figure}

\begin{figure}[h]
     \centering
    \includegraphics[width=\textwidth]{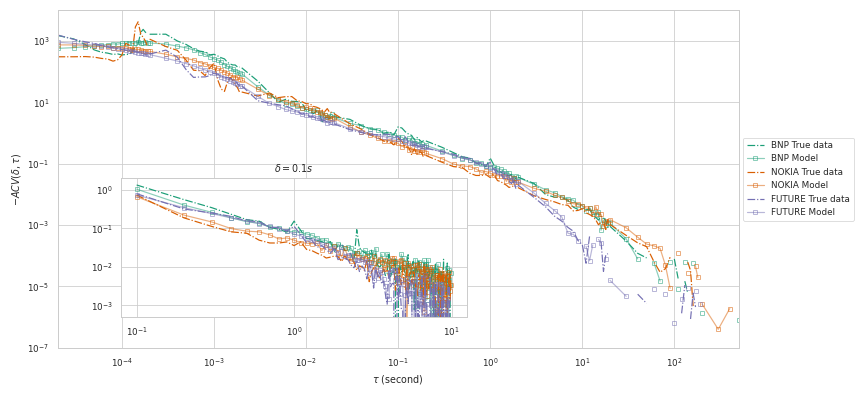}
    \caption{\textbf{Autocovariance of spread increments}. The $-ACV(\delta,\tau)$ function for the different assets. For each asset the curves that are displayed (in the main plot and in the inset) follow the same protocol as the one of Fig. \ref{fig:axa_acv_colors}}
    \label{fig:all_assets_acv}
\end{figure}
\end{appendices}
\end{document}